\documentclass[13pt]{article}

\usepackage[inner=3cm,outer=3cm,top=3cm,bottom=3cm]{geometry} 
\usepackage[utf8]{inputenc}
\usepackage{amsmath} 
\usepackage{amsthm} 
\usepackage{amssymb} 
\usepackage{mathtools} 
\usepackage{tikz} 
\usepackage{graphicx} 
\usepackage{tikz-cd} 
    \usetikzlibrary{cd}
\usepackage{mathrsfs} 
\usepackage{array} 
\usepackage{setspace} 
\usepackage{xcolor} 
\usepackage{harvard} 
\usepackage{enumitem} 
\usepackage{tcolorbox} 
\usepackage{float}
\usepackage{pdfpages} 
\usepackage{hyperref} 

\hypersetup{
    colorlinks=true,
    linkcolor={blue!50!black},
    citecolor={blue!50!black},
    urlcolor={blue!50!black} 
}

\theoremstyle{definition}
\newtheorem{definition}{Definition}
\theoremstyle{plain}
\newtheorem{lemma}{Lemma}
\theoremstyle{plain}
\newtheorem{theorem}{Theorem}
\theoremstyle{plain}
\newtheorem{proposition}{Proposition}
\theoremstyle{definition}
\newtheorem{example}{Example}
\theoremstyle{definition}

\theoremstyle{plain}
\newtheorem{corollary}{Corollary}
\newtheorem{remark}{Remark}

\newcommand{\sst}{\, | \,} 
\newcommand{\alloc}{\mathcal{A}}
\newcommand{\salloc}{\mathcal{S}}
\newcommand{\pp}{\mathscr{P}}
\newcommand{\obj}{\mathcal{O}}

\DeclareMathOperator*{\argmax}{arg\,max}

\title{Incentives and Efficiency in Constrained Allocation Mechanisms\thanks{We thank Samson Alva, Simon Board, Kim Border, Ben Brooks, Haluk Ergin, Satoshi Fukuda, Thomas Gresik, Yuhta Ishii,  Yuichiro Kamada, Timothy Kehoe, Rohit Lamba, Jacob Leshno, Jay Lu, Delong Meng, Moritz Meyer-ter-Vehn, Mich\`{e}le M\"{u}ller, Roger Myerson, Farzad Pourbabaee, Wenfeng Qiu, Doron Ravid, Phil Reny, Tomasz Sadzik, Chris Shannon, David Rahman, Ron Siegel, Ran Shorrer, Hugo Sonnenschein, Alexander Westkamp, Bill Zame and various seminar participants for helpful feedback.}}

\author{Joseph Root\thanks{Department of Economics, University of Chicago, Email: jroot@uchicago.edu} \and David S.\ Ahn\thanks{Olin Business School, Washington University in St.\ Louis. Email: ahnd@wustl.edu}}

\date{\date{\today}}

\begin{document}

\maketitle

\begin{abstract}
We study private-good allocation under general constraints. Several prominent examples are special cases, including house allocation, roommate matching, social choice, and multiple assignment. Every individually strategy-proof and Pareto efficient two-agent mechanism is a ``local dictatorship.’’ Every group strategy-proof $N$-agent mechanism has two-agent marginal mechanisms that are local dictatorships. These results yield new characterizations and unifying insights for known characterizations. We find all group strategy-proof and Pareto efficient mechanisms for the roommates problem. We give a related result for multiple assignment. We prove the Gibbard--Satterthwaite Theorem and give a partial converse. We also apply our characterization to task allocation and network regulation problems.
\end{abstract}

\section{Introduction}

Market design often involves constraints. School choice assignments must meet quotas for underrepresented students at high-performing schools. Medical residency assignments must place enough doctors in rural areas. The allocation of radio frequency in spectrum auctions must satisfy a variety of complicated engineering conditions to minimize cross-channel interference. 

Although successful approaches have been tailored for particular constraints in specific problems, to date there is little general understanding of how constraints affect the two classic economic considerations of efficiency and incentives. Theoretically, a unified approach would enable analytical insights to be shared between contexts. Practically, a flexible theory of constraints for market design would expand applicability. Real-world problems involve many ad hoc considerations that are difficult to anticipate. The tools of market design should be general enough to accommodate such constraints.

We study object allocation with private values for completely general constraints. A finite number of objects are allocated to a finite number of agents and an arbitrary constraint limits the set of feasible social allocations. Each agent has strict preferences over the objects assigned to her, but is indifferent to others' assignments. 

While other agents' assignments have no direct effect on one's well-being, others' assignments do limit the profiles of allocations that are jointly feasible. Obviously, the assignment of a house to one agent precludes another agents' consumption of that same house. Even with purely private values, constraints introduce linkage across agents' allocations. Each agent $i$ is indirectly concerned with any other $j$'s assignment, not because $i$ cares about $j$'s consumption, but rather because $j$'s assignment will limit the set of objects for $i$ that are jointly feasible with $j$'s assignment. Our goal is to study the set of incentive compatible and efficient mechanisms for a fixed, but arbitrary constraint. We study how different features of a constraint can make it amenable for implementation. For any constraint on the set of feasible allocations, our findings characterize the class of mechanisms that are immune to manipulation by any group of agents yet still yield Pareto efficient outcomes.

We start by considering two-agent environments. This case admits a surprisingly parsimonious characterization of the set of individually strategy-proof and Pareto efficient mechanisms for all constraints. We show that all such mechanisms are ``local dictatorships'' where the set of infeasible allocations is partitioned into two regions and each region is assigned a local dictator. For a given preference profile, the agents' top choices determine some (possibly infeasible) social allocation. If this allocation is feasible, then there is no conflict of interest and the mechanism assigns each agent their favorite object. Otherwise, the favorite allocation is infeasible, and a local dictator is empowered at that infeasible allocation. The dictator is assigned their top object and the non-dictator is assigned their favorite object compatible with the dictator's top object. However, not all dictatorship partitions will be strategy-proof. Instead, some structure is required of the partition to ensure these desiderata are maintained. If two infeasible allocations agree on either coordinate, that is, if they assign either agent the same object, then they must share the same local dictator. We call local dictatorship mechanisms satisfying this requirement, ``adapted local dictatorships."

With three or more agents, the set of individually strategy-proof and Pareto efficient mechanisms no longer admits a tidy characterization. Even for the classical house allocation problem, the full collection of such mechanisms remains unknown. Nevertheless, substantial progress can be made by strengthening the incentive requirement to group strategy-proofness, which requires that no coalition of agents can jointly misreport their preferences to obtain a strictly better outcome for all its members.    

Group strategy-proofness has several conceptual and analytical advantages over individual strategy-proofness. First, it is equivalent to Maskin monotonicity, a normatively appealing condition stating that if a particular allocation is chosen at some preference profile, it should remain chosen whenever every agent ranks that allocation weakly higher. Second, the additional restriction imposed by group strategy-proofness is modest: it merely rules out bossiness, situations in which an agent can alter others’ allocations without changing her own. Finally, by a result of \citeasnoun{Alva17}, though group strategy-proofness rules out any coalitional misreport, it is sufficient to rule out misreports by individuals and pairs.

This equivalence allows us to leverage our two-agent analysis. Fixing any two agents and any preference profile of the remaining agents defines a marginal constrained two-agent allocation problem. A mechanism is group strategy-proof if and only if each of these marginal mechanisms is an adapted local dictatorship. This is still not a closed-form characterization of all admissible mechanisms, since it requires a careful inspection of derived marginal mechanisms. While it is not immediately evident how these pairwise characterizations combine into a global mechanism, the finding nonetheless provides a useful method of attacking general problems. By examining the structure of the two-agent marginals, we can identify the class of group strategy-proof and Pareto efficient mechanisms for a wide range of canonical problems, as we demonstrate in specific applications.

Our study of constrained allocation yields some surprising theoretical insights. Several prominent problems which, at first glance may appear unconstrained and unrelated, can be neatly expressed as special constraints of our model. For example, the classical social choice problem corresponds to the constraint where all agents are required to consume the same object.\footnote{The term ``object" is figurative. In social choice, the objects are usually policy choices or political candidates.} From this perspective, the social choice problem is a special constrained private-goods allocation problem. A simple application of our results gives the Gibbard--Satterthwaite Theorem: that all strategy-proof social choice mechanisms are dictatorial.\footnote{For the social choice constraint, a strategy-proof mechanism is automatically group strategy-proof.} We can also formulate and prove an inverse to Gibbard--Satterthwaite: we give general conditions on the constraint which guarantee that there are non dictatorial mechanisms. We get a possibility result for two-sided matching and college assignment as immediate corollaries. 

Another prominent case of our theory is house allocation, where a finite number of indivisible objects must be assigned to agents with unit-demand. Expressed this way, the house allocation problem is almost the opposite of the social choice problem: no two agents can be assigned the same object. \citeasnoun{PyUn17} provided a full characterization of the group strategy-proof and Pareto efficient house allocation mechanisms. They showed that all such mechanisms are variants of the hierarchical exchange mechanisms of \citeasnoun{Papai00} where some agents can ``broker" objects and when exactly three agents remain, a ``braid" can form \cite{Bade16}. We generalize this problem and consider the case of combinatorial assignment where agents can be assigned bundles of up to $k$ goods. House allocation is the special case where $k=1$. We show that for $k\geq2$ the theory collapses and the only group strategy-proof and Pareto efficient mechanisms are sequential dictatorships.\footnote{Sequential dictatorship is a simple variant on the well-known serial dictatorship mechanism. In serial dictatorship, agents are called in a pre-specified order to choose their favorite object compatible with the choices of earlier dictators. In sequential dictatorship, the ordering of dictators can be endogenous to the choices of earlier dictators.} 

A third prominent problem that can be expressed as a constraint is the roommates problem, where a number of agents need to be matched into pairs. In this case, the ``objects'' are the other agents and the constraint requires that: first, no agent is matched to herself; and second, if $i$ is assigned to $j$, then $j$ is also assigned to $i$. We use our results to demonstrate that all group strategy-proof and Pareto efficient roommate mechanisms are sequential dictatorships. 

We also consider a task allocation problem. A finite number of tasks need to be allocated to a set of workers. Workers can be assigned to multiple tasks and tasks can be assigned to multiple workers. The constraint is that each task must be assigned to at least one worker. In this case, all group strategy-proof and efficient mechanisms are what we call ``unilateral." When the workers top-rank an infeasible allocation, that is, an allocation where not all tasks are assigned, a single worker is required to take-on all remaining tasks. However, this worker can be highly dependent on the reported preferences.

These examples illustrate a key conceptual contribution of our paper: to provide a novel framework unifying positive and negative results across these applications, tying together seemingly disparate environments by viewing them as different constraints on the image rather than through restrictions of preferences on the domain. Traditionally, positive results in specific environments are seen as escaping the impossibilities of Arrovian social choice by restricting preferences in the {\em domain} of the mechanism to convenient special cases, such as assuming single-peaked rankings or quasi-linear preferences. In contrast, our model can provide a different reconciliation of these positive results by interpreting these environments as relaxing constraints in the {\em image} of the mechanism: outside of the restrictive social choice constraint, all agents need not consume the same object and instead there is room for compromise to yield mechanisms beyond dictatorship. That is, our model illuminates that the  ``diagonal'' constraint implicit in the social choice problem generates maximal tension between efficiency and incentives, while other constraints allow more scope for their coexistence.

\subsection{Literature Review}

To our knowledge, this paper is the first to study the entire set of mechanisms that satisfy criteria regarding incentives and efficiency for general constraints. However, we mentioned that several canonical problems can be parameterized as specific constraints in our model, so we first review existing findings for these problems. One example is the house allocation problem, where no two agents can share the same object. The two famous group strategy-proof and Pareto efficient mechanisms for house allocation are serial dictatorship and top trading cycles, attributed to David Gale by \citeasnoun{ShSc74} and shown to have these features by \citeasnoun{Bird84}.\footnote{Serial dictatorship is analyzed comprehensively by \citeasnoun{Svensson94} and \citeasnoun{Svensson99}.} \citeasnoun{AbSo99} and \citeasnoun{Papai00} construct additional classes of group strategy-proof and Pareto efficient mechanisms that include mechanisms beyond top trading cycles and serial dictatorship. However, a complete characterization was outstanding until \citeasnoun{PyUn17} provided an impressive full description of all group strategy-proof and Pareto efficient mechanisms. \citeasnoun{PyUn17} showed that this class includes the hierarchical exchange mechanisms of \citeasnoun{Papai00}, the ``braid" mechanisms of \citeasnoun{Bade16} and a new variant of these mechanisms with ``brokers."  Our paper was inspired in part by the possibility of attaining a characterization for such an important problem.

The roommates, or one-sided matching, problem can be viewed as a further restriction on the house allocation constraint. Here, the objects are the agents, and the allocation must satisfy the constraint that if $j$ is allocated to $i$, then $i$ is also allocated to $j$. While incentives and efficiency are  well-understood for house allocation, similar insights for one-sided matching were yet unknown. This is in large part because the roommates problem may have no stable outcomes, as originally observed in the famous paper by  \citeasnoun{GaSh62}. Since then, a voluminous literature in operations research and computer science, starting with \citeasnoun{Irving85}, constructs efficient algorithms to find stable matchings when they exist, and the study of stability for one-sided matching is now well-known as the ``stable roommates problem.'' In contrast, there is little discussion of incentives and efficiency for the roommates problem.\footnote{An exception is the working paper by \citeasnoun{AbMa04} that studies the computational hardness of finding Pareto optimal matches for the roommates problem.} As an application of our main results, we find that all group strategy-proof and Pareto efficient mechanisms for the roommates problem are sequential dictatorships. To our knowledge, this observation is novel and establishes a result for one-sided matching akin to the characterization of Gibbard and Satterthwaite for social choice or that of \citeasnoun{PyUn17} for house allocation. 

A relaxation of house allocation allows agents to own multiple houses, which is sometimes called the combinatorial or multiple assignment problem. Our main finding for this problem is that sequential dictatorship is the unique group strategy-proof and Pareto efficient mechanism for the environment where each agent can be assigned at most $k\geq 2$ objects. This mirrors related results for multiple assignment. To our knowledge, \cite{Papai01} provided the first such result. She proves this for the case where $k$ is maximal, that is equal to the number of objects, while we consider the general case. \citeasnoun{Hatfield09} provides this characterization for the case where $k$ is a lower bound on the number of objects, as opposed to the upper bound we consider here.

A last important constraint is the classic Arrovian social choice model. The celebrated theorem of  \citeasnoun{Gibbard73} and \citeasnoun{Satterthwaite75} initiated the field of implementation theory by observing the tension between incentives and efficiency for social choice. In our model, the Arrovian social choice corresponds to the case where all agents must be assigned a common outcome. We derive the Gibbard-Satterthwaite Theorem as a corollary of our main characterization. This provides a new perspective on the classic result by considering social choice environments as a severe constraint in the allocation problem with purely private values.

Our general environment with private goods was examined by \citeasnoun{BaBeMo16} for social choice. They study implementation across different restrictions on the {\em domain} of preference. In contrast, we always consider the full domain of preferences. We complement the insights of \citeasnoun{BaBeMo16} by considering different constraints on the {\em image} of allocations that are feasible for a mechanism. Related, a series of papers consider constraints in social choice problems where the common consumption space is a product space \cite{BaMaNe97,BaMaSe98,BaMaNe05}, adding feasibility considerations to the general model of \citeasnoun{BoJo83}. These papers also connect the structure of the feasible set to the space of mechanisms satisfying various normative axioms, so our exercise of connecting constraints to mechanisms has precedents. An important difference is that these papers allow preferences to depend on the entire allocation profile, rather than the pure private values we assume in this paper.

\citeasnoun{Sonmez99} observed that problems such as two-sided matching, house allocation, and the exchange of indivisible resources can all be formulated as allocation problems by limiting the set of allocations. He employs this framework to study environments with initial endowments, relating the existence of individually rational, Pareto efficient, and strategy-proof mechanisms to properties of the core. By contrast, our setting does not feature endowments, so notions such as individual rationality or the core cannot be meaningfully defined. In this sense, our results are complementary: they apply to collective ownership environments, where the set of objects can be allocated in any feasible manner.

Finally, in contemporaneous work \citeasnoun{Meng19} proved a theorem for social choice with exogenous indifference classes which is mathematically equivalent to our two-agent characterization. While indifference classes and constraints can be mapped to each other in a mathematical sense, our aims are different and the contributions of the two papers are independent. \citeasnoun{Meng19} does not observe the usefulness of the result for constrained market design, while we did not observe its usefulness for problems with indifferences. We compare the technical results more specifically when we introduce our two-agent characterization. 

\section{Preview with Two Agents} \label{sec: two agents introduction}

We begin by previewing the scope of our results for the two-agent case. Each agent $i = 1,2$ must be allocated some object $a_i \in \obj$, where $\obj$ is a set of objects. This case is particularly simple to present because the space of allocations $\obj \times \obj$ is a square that can be easily visualized and our results are the simplest for this case, while still being rich enough to capture a variety of interesting problems. Only a subset of these allocations are feasible, and denote the subset of feasible allocations by $C \subseteq \obj \times \obj$. Let $\bar{C} = \obj \times \obj \setminus{C}$ denote the set of infeasible allocations. For now, we assume, for convenience, that the constraint is rich enough so that for every object $x \in \obj$, there is some feasible allocation where $a_i = x$.\footnote{This assumption is not necessary and is dropped in the main model.} 

As mentioned in the Introduction, several canonical implementation problems for two agents can be expressed as specific constraints on allocations. In addition, we can consider some interesting novel problems. The related constraints can all be visualized in the grid of allocations.

\begin{figure}[h]
\centering
\includegraphics[scale=0.6]{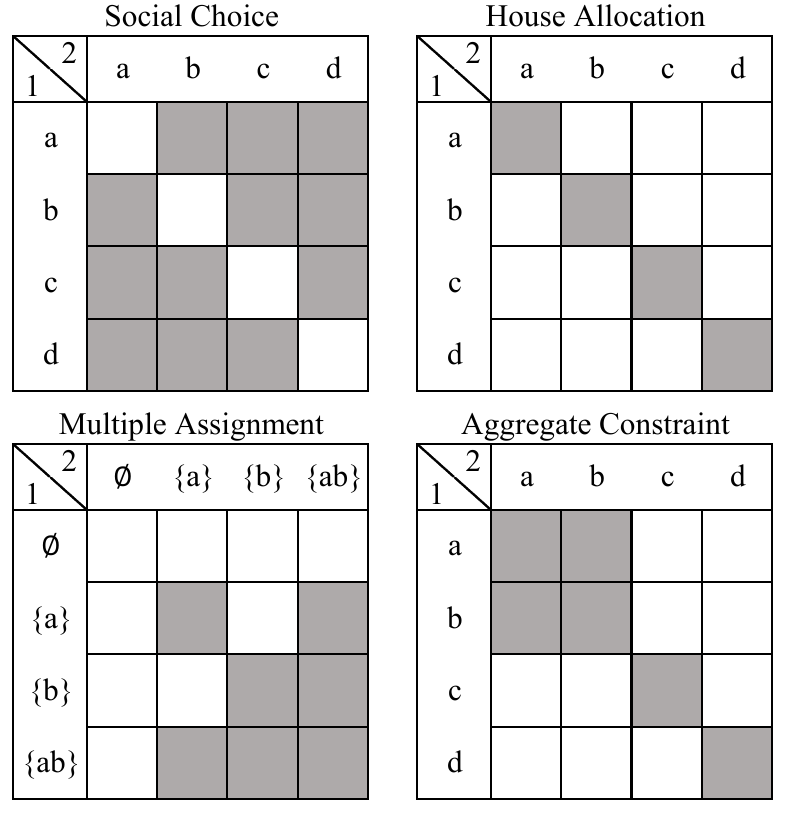}
\caption{Four two-agent examples. The infeasible allocations are shaded gray. Agent $1$'s allocation is determined by the row and $2$'s allocation is determined by the column.}
\label{two-agent examples}
\end{figure}

\begin{itemize}
    \item Social choice: Arrovian social choice decides a single social outcome for all agents, like whether to build a bridge or which candidate to elect. In our language, this means that agents share a common object. This is equivalent to the constraint $C = \{(a_1, a_2) \in \obj \times \obj: a_1 = a_2\}$. Visualizing all allocations as a square, social choice requires that the allocation live on the diagonal. The first panel in Figure \ref{two-agent examples} illustrates, where the feasible diagonal is in light squares, while the infeasible allocations are shaded by dark squares.
    \item House allocation: Here objects are houses. Each agent is assigned a house, and a house can be owned by at most one agent. The corresponding constraint is $C = \{(a_1, a_2) \in \obj \times \obj: a_1 \neq a_2\}$. With two agents, the house allocation constraint is the complement of the social choice constraint. This is shown in the Figure \ref{two-agent examples}, where in this case the only infeasible allocations are on the diagonal. 
    \item Multiple assignment: This problem extends house allocation by allowing each agent to own multiple objects. Letting $H$ be the set of houses, set $\obj = 2^H$. Each agent is allocated a set of houses $s_i \subseteq H$, and each house can be allocated to at most one agent. Then $C = \{(s_1, s_2) \in \obj \times \obj : s_1 \cap s_2 = \emptyset \}$.\footnote{In the general multiple assignment problem, additional structure, like a provision that no agent own more than $k$ houses, can be imposed. For now, we consider the simpler problem without this additional structure to more cleanly explain the two-agent intuition.} Figure \ref{two-agent examples} illustrates where now the objects are subsets of two houses $a,b$ rather than individual houses. For example, the allocation $(\{a\}, \{a,b\})$ is infeasible because both agents own $a$, and is thus part of the dark infeasible area.
    \item Allocation with aggregate constraints: This is a variation of (single-object) house allocation. Suppose that some agent must be allocated an object from a subset $\mathcal{Q} \subseteq \obj$. This is captured by the constraint $C = \{(a_1, a_2) \in \obj \times \obj: a_1 \neq a_2 \mbox{ and } \{a_1, a_2\} \cap \mathcal{Q} \neq \emptyset \}$. For example, the agents could be an economics department and a business school recruiting from the same pool of candidates $\obj$. However, the provost does not want both schools to hire a senior candidate for budgetary reasons. Then the set $\mathcal{Q}$ is junior candidates, since at least one division must hire a junior professor. Another interpretation is that objects are residency assignments for medical students and the set $\mathcal{Q}$ is hospitals in an underserved rural area. This is a classic design consideration, where one resident must be allocated to a hospital in that rural area. The case where $\mathcal{Q} = \{c,d\}$ is illustrated Figure \ref{two-agent examples}. In addition to the diagonal being infeasible, the allocations $(a,b)$ and $(b,a)$ are also shaded because neither agent receives $c$ or $d$.
\end{itemize}

A mechanism takes the agents' strict preferences over objects $(\succsim_1, \succsim_2)$ and outputs a feasible allocation $f(\succsim_1, \succsim_2) \in C$. Our first observation for two agents will be that every strategy-proof and Pareto efficient mechanism is parameterized with a bi-partition $D : \bar{C} \rightarrow \{1,2\}$ of the infeasible allocations into two disjoint regions, one where agent 1 gets his favorite object and another where agent 2 does. Suppose the top choices for the agents' preferences are $(a_1, a_2)$. If feasible, an efficient mechanism will obviously output this allocation. The nontrivial case is what happens when there is a conflict of interests and the top choices are infeasible, $(a_1, a_2) \notin C$. Then $D(a_1,a_2)$ is a ``local dictator'' when $(a_1,a_2)$ is the pair of top objects. The other agent compromises, and is assigned her favorite choice among those objects that are compatible with the local dictator's top choice. Part of the bite in our result is that the assignment of local dictator can only depend on the top-ranked objects in the preference profile, and is invariant to the ordering of lower-ranked objects. So we can picture every (individually) strategy-proof and Pareto efficient mechanism as an assignment of local dictators to infeasible allocations.\footnote{With two agents, individual strategy-proofness and Pareto efficiency are equivalent to group strategy-proofness when the mechanism has full range.}

For concreteness, Figure \ref{local dictatorship example} illustrates a specific local dictatorship for a different constraint. The dark squares represent the infeasible allocations $\bar{C} = \{ (a,b); (b,a); (b,c); (c,c) ; (d,a)\} $. The local dictator when $(a,b)$ are the top choices is agent 1, $D(a,b) = 1$, and the local dictator at every other infeasible allocation is agent $2$, $D(b,a) = D(b,c) = D(c,c) = D(d,a) = 2$. Consider the preference profile where  $a \succ_1 b \succ_1 c \succ_1 d$ and $b \succ_2 d \succ_2 a \succ_2 c$. Then the top objects are $(a,b)$, which is infeasible. Agent 1 is the local dictator, so agent 2 must move to her top choice that is feasible when $a_1 = a$, which is $2$'s second-favorite object $d$. So the mechanism outputs the allocation $f(\succsim_1, \succsim_2) = (a,d)$ for this profile, moving across the row of top choices to a feasible allocation because agent $1$ gets to keep object $a$. Now consider the preference profile where $c \succ'_1 b \succ'_1 d \succ'_1 a$ and $c \succ'_2 a \succ'_2 b \succ'_2 d$. Then the top objects are $(c,c)$ and $2$ is the local dictator here. So $1$ must compromise. Agent $1$'s next-favorite object is $b$, but that is not feasible with $a_2 = c$ because $(b,c)\notin C$. So Agent 1 most move further down her preference to object $d$, which is jointly feasible with $a_2 = c$. So the mechanism outputs $f(\succ'_1, \succ'_2) = (d,c)$, moving across the column of top choices because agent $2$ gets to keep object $c$ here.

\begin{figure}[h]
\centering
\includegraphics[scale=0.7]{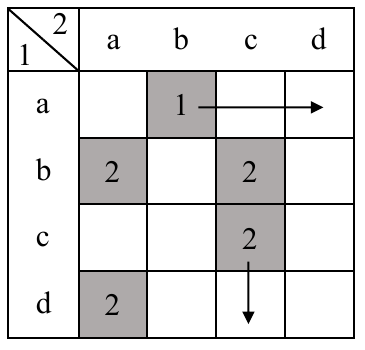}
\caption{Local Dictatorship Example}
\label{local dictatorship example}
\end{figure}

However, there are requirements on the assignment of local dictators, that is, not all local dictatorship assignments yield strategy-proof and Pareto efficient mechanisms. It turns out that to maintain strategy-proofness, the same local dictator must be assigned to any two infeasible allocations that agree on one dimension. Graphically, this means that every column and every row of infeasible allocations must share the same local dictator. We will say such local dictatorship assignments are ``adapted.''  Since equality is transitive, that means any two infeasible allocations that are connected to each other through only horizontal movements through rows or vertical movements through columns must share the same local dictator. However, two infeasible allocations that are not connected in this manner may have different local dictators.

Returning to the example in Figure \ref{local dictatorship example}, this particular assignment is adapted. There are two connected components. One is the single infeasible allocation $\{(a,b)\}$, which is disconnected from the rest of the infeasible space since no other infeasible allocation shares the same row or column. The second component consists of the other four infeasible allocations $\{(b,a); (b,c); (c,c); (d,a)\}$. All of these squares must have the same local dictator. This is because, $(d,a)$ must have the same local dictator as $(b,a)$ since they are on the same column, $D(d,a) = D(b,a)$, which must have the same local dictator as $(b,c)$ on the same row, $D(b,a) = D(b,c)$, which must have the same local dictator as $(c,c)$ on the same column, $D(b,c) = D(c,c)$). In this local dictatorship assignment, $1$ is the local dictator in one component while $2$ is the local dictator in the other. 

This requirement restricts the number of available mechanisms. In particular, if every infeasible allocation can be connected to every other by moving only vertically or horizontally, then the same agent must be the local dictator no matter what the announcements are. In these cases, all strategy-proof and Pareto efficient mechanisms are dictatorial.

For example, revisiting the social choice constraint, the structure of the infeasible set requires that the same agent be assigned to all infeasible allocations. For example, $(a,b)$ and $(c,b)$ are on the same column so must have a common assignment $D(a,b) = D(c,b)$. Similarly, $(c,b)$ and $(c,d)$ are on the same row, so $D(c,b) = D(c,d)$. Then, combining these two restrictions, $(a,b)$ and $(c,d)$ must also share the same dictator. Graphical inspection of the constraint confirms that any two infeasible allocations are connected in this manner. So there are only two strategy-proof and Pareto efficient mechanisms: one where agent $1$ is the local dictator everywhere, and another where agent $2$ is. This means that the only strategy-proof and Pareto efficient mechanisms for social choice are dictatorial. Therefore, an immediate corollary of our two-agent characterization is the two-agent version of the Gibbard--Satterthwaite Theorem.

The only case of a social choice constraint where the infeasible space is not connected in this manner is with two objects, $a$ and $b$. Then the two infeasible allocations are $(a,b)$ and $(b,a)$. In this case there are four admissible mechanisms, the two dictatorships and two non-dictatorial ones. For example, suppose $D(a,b) = 1$ and $D(b,a) = 2$. Under this mechanism, if agent $1$ wants the outcome to be $a$ but agent $2$ wants $b$, then the top choices are $(a,b)$ and $1$ is the local dictator ($D(a,b) = 1$). The only feasible allocation consistent with $1$ getting her favorite object $a$ is $(a,a)$, so the allocation here must be $(a,a)$. On the other hand, if $(b,a)$ are the top objects, then $2$ is the dictator and again the outcome must be $(a,a)$ because $2$ gets her favorite object $a$ as the local dictator. In other words, this local dictatorship assignment is equivalent to the unanimity rule with object $a$ as the default option or status quo in case of a tie. The other non-dictatorial mechanism has the local dictators reversed, which is equivalent to the unanimity rule with object $b$ as the default. So not only does our approach produce the Gibbard--Satterthwaite Theorem, but it is capable of demonstrate what is special about the two outcome case: the constraint is no longer connected. 

Now consider the house allocation constraint. This is opposite of the social choice constraint: no infeasible allocation shares a row or column with any other infeasible allocation, since they all live on the diagonal. Then there are $2^{|\obj|}$ admissible mechanisms, in the case with four objects that is $2^4 = 16$ mechanisms. In fact, the house allocation problem has the maximal number of disconnected components of the infeasible space, so is the constraint that admits the most mechanisms. When $1$ is the local dictator everywhere, this defines serial dictatorship where $1$ receives her favorite house and $2$ gets her second-favorite house in case of a conflict. If local dictatorship is distributed, say $D(a,a) = D(b,b) = 1$ and $D(c,c) = D(d,d) = 2$, this defines a system of implicit property rights, in this case where agent $1$ has claims on houses $a$ and $b$, and agent $2$ has claims houses $c$ and $d$. With two agents, this system of property rights is equivalent to top-trading cycles where the agent's endowments are respectively $\{a,b\}$ and $\{c,d\}$.

Next examine the multiple assignment problem. Here, visual inspection verifies that the entire infeasible space is connected through horizontal and vertical moves through rows and columns. So, as with the case of social choice, the only admissible mechanisms are serial dictatorships. Here, the universal dictator gets her favorite set of objects, and the residual agent receives her favorite subset of the remaining objects unclaimed by the dictator. There is a sharp divergence between single-object allocation which has the maximal number of admissible mechanisms and multiple-object allocation which has the minimal number. Our approach provides a clean graphical intuition for this departure.

Finally, recall the last example of an allocation problem with a simple aggregate constraint, where either $c$ or $d$ must be allocated to some agent. Again by visual inspection, we conclude there are three connected components of the infeasible space: The block $\{ (a,a), (a,b), (b,a) , (b,b)\}$ (in the upper-left of the illustration), the single allocation $\{(c,c)\}$, and the single allocation $\{(d,d)\}$. So for this specific case there are $2^3 = 8$ mechanisms. Much of the intuition for these mechanisms is as for the standard house allocation problem, the issue is that agents' claims must have additional structure to satisfy the aggregate constraint. Specifically, the same agent must have rights over both $a$ and $b$, while in the standard house allocation problem, different agents could have claims on either. 

More important than the exact findings for specific applications, is the framework we introduce that unifies these findings in an omnibus model that yields social choice, house allocation, and multiple assignment as special cases. Each of these is a special constraint, and we can identify the space of mechanisms with good incentive and efficiency properties from the structure of the infeasible allocations. The approach presents a visual connection between the social choice problem and the multiple assignment problem: both have completely connected infeasible spaces, so only admit serial dictatorships. As we will see in the general model, other canonical problems like two-sided matching and the roommates problem that are trivial with two agents can also be expressed as special constraints with more agents.

\section{Model}

We begin by introducing primitives. Let $N$ be a finite set of \textbf{agents} and $\obj$ be a finite set of \textbf{objects}. We use the term ``object'' because our leading examples are allocation problems, but note that $\obj$ need not be physical objects like houses, but can be political candidates, roommates, and so on. Let $\alloc$ denote the set of all allocations of objects to agents. We sometimes write an allocation as a tuple $(a_i)_{i\in N}$ and sometimes as a function $\mu:N\rightarrow \obj$. A \textbf{suballocation} is an allocation to a subset of agents, i.e. a function $\sigma:M\rightarrow \obj$ where $M\subset N$. Let $\salloc$ denote the set of all suballocations.
 
It is possible that not all allocations are feasible, and we denote the (nonempty) feasible set by $C\subset \alloc$. Importantly, the constraint is exogenous to the problem and is given to the mechanism designer as a fixed set of feasible outcomes. Our analysis applies to arbitrary constraints so we can incorporate settings where the set of objects differ for different agents.\footnote{For example, if each agent had their own set of objects $\obj_{i}$, we could let $\obj = \bigcup_i \obj_{i}$ and add the constraint that $i$ be assigned an object in $\obj_{i}$.} 
 
Moving to preferences and types, agents have strict preferences over the \textit{objects}. A preference for agent $i$ will typically be denoted $\succsim_i$ and we will write $x\succsim_i y$ to mean that either $x$ is strictly ranked above $y$ or $x=y$. We assume purely private goods, or selfishness over allocations. That is, the only part of an allocation $(a_i)_{i\in N}$ that matters to agent $j$ is her own allocation $a_j$, and she is indifferent between any two allocations $(a_i)_{i\in N}$ and $(a_i')_{i\in N}$ where $a_j = a'_j$. Thus any other agent's consumption imposes no direct externality on agent $j$. This does not mean there is no conflict of interest in this model. By assuming purely private values, all of the tension in our model flows only through the constraint. That is, the issue is due to limited ``supply'' of objects and not due to direct externalities. 

We will use $P$ to denote the set of strict preferences (or linear orders) on $\obj$  and $\pp=P^{N}$ to denote the set of preference profiles.\footnote{A binary relation $B\subset \obj\times\obj$ is a linear order if it is complete, transitive, and antisymmetric.} Our primary object of interest in this paper is a \textbf{feasible mechanism}, which is simply a function $f:\pp\rightarrow C$. Our task will be to find feasible mechanisms satisfying desirable conditions regarding incentives and efficiency to be introduced below. 

\subsection{Examples and applications}

In Section \ref{sec: two agents introduction}, we gave several examples of constrained allocation problems including the social choice, house allocation, and multiple assignment problems. Here we list some others prominent examples:
\begin{itemize}
\item School Choice: A finite number of students $N$ need to be assigned to one of a finite number of schools $A$. Each school $a$ has capacity $q_a$. One school $\emptyset$ corresponds to the option to remain unmatched and $q_\emptyset = N$. This gives rise to the constraint $$C=\{\mu:N\rightarrow A \sst \vert \mu^{-1}(a)\vert \leq q_a \text{ for all }a\in A  \}.$$ 
\item Roommates Problem: Universities are often tasked with assigning students into shared dormitory rooms. Assuming $N$ is even, this problem can be captured in our environment by setting $\obj=N$  and imposing the constraint 
$$C=\{\mu:N\rightarrow N \sst \mu\circ \mu=id \text{ and } \mu(i)\neq i \text{ for all }i \}$$
where $id$ is the identity map $i\mapsto i$. The first condition requires that if $i$ is assigned roommate $j$ then $j$ is also assigned $i$ and the second condition requires that all agents are assigned a roommate.
\item Two-sided Matching: The set of agents $N$ is composed of two disjoint sets $M$ and $W$ where $\vert M \vert = \vert W\vert$. Agents need to be matched into pairs with the constraint that $m$'s need to be matched with $w$'s. This gives rise to the constraint 
$$C=\{\mu:N\rightarrow N \sst \mu\circ \mu=id \text{ and } \mu(m)\in W \text{ for all }m \}$$


\end{itemize}

\subsection{Notation}

Before moving on, we record here some notation that will be used throughout the paper. For any subset $M\subset N$, given a preference profile $\succsim=\left(\succsim_{i}\right)_{i\in N} \in \mathscr{P}$ and a profile of alternative preferences for agents in $M$, $(\succsim'_{j})_{j\in M}$, we will write $(\succsim'_{M},\succsim_{-M})$ to refer to the profile in which an agent $j$ from $M$ reports $\succsim'_{j}$ and any agent $i$ from $N-M$ reports $\succsim_{i}$. We will often want to consider how a mechanism $f$ changes when a few agents change their preferences, that is the difference between $f(\succsim)$ and $f(\succsim'_{M},\succsim_{-M})$. When the initial preference profile $\succsim$ is clear, we will sometimes write $\succsim_{-}$ instead of $\succsim_{-M}$. For any set of agents $M$, let $\pi_{M}:\alloc \rightarrow \obj$ be the projection map so that given an allocation $(a_{j})_{j\in N}$, we have $\pi_{M} a\coloneqq(a_{j})_{j\in M}$. Given a constraint $C\subset \mathcal{A}$ and a subset of agents $M\subset N$, let $C^{M}=\{\mu:M\rightarrow \mathcal{O} \sst \exists b\in C\text{ s.t. } b_{i}=\mu(i) \, \forall i\in M \}=\pi_M(C)$ which we will call the \textbf{projection of }$C$ \textbf{on} $M$. An element of $C^{M}$ will be referred to as a \textbf{feasible suballocation} for agents in $M$. If $\mu:M\rightarrow \mathcal{O}$ and $\mu':M'\rightarrow \mathcal{O}$ are suballocations with $M\subset M'$ which agree on their shared domain, $\mu'$ is called a \textbf{extension} of $\mu$. If $\mu'$ is a feasible suballocation (which of course implies that $\mu$ is) then $\mu'$ is called a \textbf{feasible extension} of $\mu$. If $\mu'$ assigns an object to each agent, it is called a \textbf{complete extension} of $\mu$. Given a feasible suballocation $\mu$, we will let $C(\mu)$ denote the set of complete and feasible extensions of $\mu$.  

For $x\in \obj$ and $\succsim_{i}\in P$, define $LC_{\succsim_{i}}(x)=\{y\in \obj \sst y\prec_{i} x\}$ be the (strict) \textbf{lower contour set} of $x$ at $\succsim_{i}$. Likewise, $UC_{\succsim_{i}}(x)=\{y\in \obj \sst y\succ_{i} x\}$ is the (strict) \textbf{upper contour set} of $x$ at $\succsim_{i}$. For a preference $\succsim_{i}$, define $\tau_{n}(\succsim_{i})$ as the $n$th top choice under $\succsim_{i}$. Likewise, for any preference profile $\succsim$, define $\tau_{n}(\succsim)$ as the allocation in which each agent gets their $n$th top choice. To save on notation, we will often omit the subscript when referring to the top choice (i.e. writing $\tau(\succsim)$ to mean $\tau_{1}(\succsim)$). We will use $\bar{C}$ to denote the set of infeasible allocations. If $A$ and $B$ are sets of objects and $\succsim\in P$, we say $A\succsim B$ if $a\succsim b$ for all $a\in A$ and $b\in B$. For disjoint sets of objects $A_{1},A_{2}\dots A_{m}$ we will denote $P\left[A_{1},A_{2}\dots A_{m}\right]=\{\succsim\in P \sst A_{1}\succ A_{2} \succ \dots \succ A_{m} \}$ and $P^{\uparrow}\left[A_{1},A_{2}\dots A_{m}\right]=\Big\{\succsim\in P \sst A_{j}\succ \obj \setminus \bigcup_{i=1}^{j}A_{i} \text{ for all }j\Big\}$. In words, $P\left[A_{1},A_{2}\dots A_{m}\right]$ is the set of preference profiles where all objects in $A_k$ are preferred to all objects in $A_l$ if $k<l$ but the ranking of objects in $\obj-\cup A_k$ is left unrestricted. $P^{\uparrow}\left[A_{1},A_{2}\dots A_{m}\right]$ is the set of preferences where the objects in $A_1$ are ranked above all other objects, the objects in $A_2$, ranked above all objects other than $A_1$ and so on. When the $A_{i}$ are singletons, we will abuse notation and drop the curly brackets, writing for example $P^{\uparrow}[a]$ to denote $P^{\uparrow}[\{a\}]$. We will also abuse notation slightly and use $N$ to refer both to the set of agents and to the number of agents.

\subsection{Desiderata}

In practice, mechanisms are designed to satisfy efficiency and incentive properties, for which the following are well-known conditions.

\begin{definition}
A mechanism $f:\pp \rightarrow C $ is
\begin{enumerate}
    \item \textbf{strategy-proof} if, for every $\succsim \in \pp$ and every $i\in N$, there is no $\succsim'_i$ $$f_{i}(\succsim_i',\succsim_{-i})\succ_{i}f_{i}(\succsim_i,\succsim_{-i})$$ for all $\succsim'_{i}\in P$. That is, truth-telling is a weakly dominant strategy.
    \item \textbf{group strategy-proof} if, for every $\succsim \in \pp$ and every $M\subset N$, there is no $\succsim'_{M}$ such that
    \begin{enumerate}
    \item $f_{j}(\succsim'_{M},\succsim_{-M})\succsim_{j}f_{j}(\succsim) \text{  for all  }j\in M $;
    \item $f_{k}(\succsim'_{M},\succsim_{-M})\succ_{k}f_{k}(\succsim)$ for at least one $k\in M$.
    \end{enumerate}
    \item {\bf pairwise strategy-proof} if, for every $\succsim \in \pp$ and every pair of agents $i,j$, there is no $\succsim_i'$ and $\succsim_j'$ such that
    \begin{enumerate}
    \item $f_{i}(\succsim'_{i},\succsim_j',\succsim_{-\{i,j\}})\succsim_{i}f_{i}(\succsim)$;
    \item $f_{j}(\succsim'_{i},\succsim_j',\succsim_{-\{i,j\}})\succ_{j}f_{j}(\succsim)$.
        \end{enumerate}
    \item \textbf{Pareto efficient} if there is no allocation $a\in C$ and preference profile $\succsim$ such that $a\neq f(\succsim)$ and $a_{j}\succsim_{j}f_{j}(\succsim)$ for all $j$.
\end{enumerate}
\end{definition}

\noindent Strategy-proofness requires that no agent can improve her outcome by misreporting her preference. Group strategy-proofness is similar, but instead requires that no group can collectively misreport their preferences without hurting anyone while strictly benefiting at least one agent. This is often called ``strong group strategy-proofness'' to contrast it with a weaker notion requiring that deviating coalitions make all  agents strictly better-off.\footnote{In our setting, weak group strategy-proofness is equivalent to strategy-proofness \cite{BaBeMo16}.} Pareto efficiency is a classic efficiency axiom. It requires that there is no way to make any agent better-off without making another agent worse-off. Group strategy-proofness can be relaxed to pairwise strategy-proofness, which only requires that no pair of agents can profitably deviate. In fact, as we will see when we move to the $N$-agent case, our environment falls under the hypotheses of Theorem 1 by \cite{Alva17} so there is no gap between group and pairwise strategy-proofness.

One candidate deviating coalition is the grand coalition. So if $f$ is group strategy-proof and $f(\succsim)=a$ for some profile $\succsim$, then $a$ can never Pareto dominate $f(\succsim')$ for any other profile $\succsim'$, since all agents could collectively report $\succsim$ instead of $\succsim'$ to get an improvement. This leads to the following simple fact.

\begin{remark} \label{PE image}
If $f:\pp \rightarrow \alloc$ is group strategy-proof then it is Pareto efficient on its image.\footnote{That is, if an allocation Pareto dominates $f(\succsim)$, then that allocation is outside the image of $f$.}
\end{remark}

We will use this fact to reduce group strategy-proofness and efficiency to group strategy-proofness and surjectivity. The goal of this paper is to understand the correspondence between the primitives (the set of agents, objects, and the constraint) and the set of group strategy-proof and Pareto efficient mechanisms. We will use $GS(C)$ to denote the set of group strategy-proof and Pareto efficient mechanisms.

\subsection{Constrained Allocation and Exogenous Indifferences}

    Constrained object allocation is closely related to the social choice problem where there are exogenous indifferences as studied in \citeasnoun{sato2009strategy}, \citeasnoun{barbera2011free} and \citeasnoun{Meng19}. This problem can be described as follows. Suppose that $X$ is a finite set of outcomes. A finite set of voters $N$ have preferences over $X$, however each is assumed to have a given set of indifferences. That is, for each $i\in I$, there is a partition $X^{i}=\{X^{i}_1, \dots, X^{i}_{m_i}\}$ on $X$ so that for any $k\neq l$, we have $X^i_{k}\cap X^i_{l}=\emptyset$ and $\cup_l X^{i}_l=X$. Agent $i$ is assumed to be indifferent between any two $x$ and $y$ contained in a single $X^i_{l}$ and has a strict preferences if $x$ and $y$ belong to different $X^i_j$. Constrained object allocation can be seen as a special case by letting $X=C$ and, for each $i$, partitioning $X$ into the allocations with a fixed allocation for $i$. Going the other direction, we can also rewrite social choice with exogenous indifference classes as constrained object allocation as follows. Let $X^{i}$ be the set of indifference classes for $i$ and let $\obj$ be the disjoint union of the $X^i$ over all $i$. Let $C$ be the constraint such that an allocation $a=(a_i)_{i\in N}$ is feasible if and only if (1) for all $i$, $a_i\in X^i$ and (2) $\cap_i a_i\neq \emptyset$. That is, an allocation of indifference classes is feasible if and only if each agent is assigned one of their indifference classes and there is some outcome in $X$ that lives in the intersection of these indifference classes. From this perspective, the two problems are equivalent. In this paper, we favor the constrained allocation perspective because we feel that it is easier to formulate the special cases that we are interested in.

\section{Characterization Results}

We begin with the two-agent case where we can explicitly construct the class of strategy-proof and Pareto efficient mechanisms for an arbitrary constraint. All such mechanisms turn out to be what we call ``adapted local dictatorships'' where a dictator is chosen as a function of the announced preferences. We then turn to the general case where we introduce the notion of ``marginal mechanisms" and show that an $N$-agent mechanism is group strategy-proof if and only if each $2$-agent marginal mechanism is an adapted local dictatorship. We apply these findings in Sections \ref{sec: existence} and \ref{sec: specific constraints}.

\subsection{Two Agents}\label{two agent section}

For the special but important case with two agents, strategy-proof and Pareto efficient mechanisms take a particularly simple form. The set of infeasible allocations $\bar{C}$ splits into two subsets $D_{1}$ and $D_{2}$ over which the two agents are given dictatorship rights. Specifically, given any preference profile, we consider the allocation $(a,b)$ where $a$  and $b$ are respectively the top choices of agents $1$ and $2$. If $(a,b)$ is feasible, Pareto efficiency forces this allocation. Otherwise, $(a,b)$ belongs to either $D_{1}$ or $D_{2}$. If it belongs to $D_{1}$, agent $1$ receives her top object $a$, while $2$ receives her favorite object $b'$ compatible with $a$ (such that $(a,b')\in C$). Likewise if $(a,b)$ belongs to $D_{2}$, $2$ keeps her top choice and $1$ must compromise. We call this type of mechanism a ``local dictatorship." We dub the two agents the ``local dictator" and the ``local compromiser" respectively.\footnote{ To our knowledge, the term ``local dictatorship was first used by \citeasnoun{sato2012strategy}. That paper considers a social choice problem where the set of outcomes is partitioned into categories so that all agents have rankings that are lexicographic in the sense that categories are ranked and then alternatives within categories are ranked to break ties. He shows that mechanisms decompose into a category-level mechanism and a within-category mechanism. The within-category mechanisms can have different dictators than the category-level dictator.}

Given that we have allowed for complete generality in the constraint, the procedure above should account for the possibility that the local compromiser may not be able to find a feasible object. For example, we should not specify the local dictator at $(x,y)$ to be agent $1$, if for all $y'$, the allocation $(x,y')$ is infeasible, so there is no compromise agent $2$ could make to allow agent $1$ to consume her favorite object $x$. Notice, however that this is a trivial difficulty as $1$ can never feasibly be assigned object $x$. It would seem that we should therefore be able to ignore $1$'s ranking of $x$. This turns out to be true, and we can ignore objects that are never assigned to an agent without loss of generality. It is no more difficult to show this for any number of agents, so we include the general result here.

\begin{lemma}\label{dumb objects dont matter}
Fix a constraint $C$ for any number of agents. Let $\chi_{i}=\{y\in \obj \vert \forall y_{-i} \hspace{0.1cm}(y,y_{-i})\in \bar{C}\}$. Suppose $f:\pp \rightarrow C$ is group strategy-proof and Pareto efficient. If $\succsim$ and $\succsim'$ are preference profiles such that, for all $i$,  $\left. \succsim_{i} \right|_{\obj-\chi_{i}} = \left. \succsim'_{i} \right|_{\obj-\chi_{i}}$,  then $f(\succsim)=f(\succsim')$
\end{lemma}

 Let $\bar{C}^{*}=\{(x,y) \sst (x,y)\notin C\text{ and }x\notin \chi_{1}, y\notin \chi_{2}\}$. That is, $\bar{C}^{*}$ is the set of infeasible allocations excluding the trivial cases described above. We will call $\obj - \chi_{i}$ agent $i$'s \textbf{individually-feasible} choices. 
 
 As mentioned, Pareto efficiency requires allocating both agents their top choices if doing so is feasible. The main job of a mechanism is to adjudicate the outcome when one agent must give up on her top choice. It turns out that strategy-proofness will demand a local dictator is determined as a function of only the agents' highest ranked objects. Suppose $(a,b)$ is the allocation that assigns each agent her favorite individually-feasible object. If $(a,b) \in C$, then any efficient mechanism gives this allocation. If not, then a local dictatorship assigns an agent as dictator. That dictator, say it is agent 1, gets her favorite object $a$, while the non-dictator agent $2$ must compromise and is assigned her favorite object among those that are mutually feasible with $1$ being assigned $a$. Since the identity of the dictator only depends on the agents' top choices, a local dictatorship partitions the space of infeasible allocations into two groups, one where agent $1$ is the dictator and another for agent $2$.\footnote{Note that this does not imply local dictatorships satisfy ``tops-only'' conditions from social choice. That is because the non-dictator's object still depends on her entire rank order to determine her second-best assignment.}
 
\begin{definition}
A mechanism is a \textbf{local dictatorship} for the constraint $C$ if there is a local dictator assignment $D : \bar{C}^{*} \rightarrow \{1,2\}$ such that, for any $(\succsim_{1},\succsim_{2})$, if $a$ and $b$ are  $1$ and $2$'s top individually-feasible choices, 
 \begin{equation} \label{local dictatorship}
    f(\succsim_1,\succsim_2)=
    \begin{cases}
    (a  ,b)  &\text{ if } (a,b) \in C \\
    \left(a, \argmax_{\succsim_2} \{ b' : (a,b') \in C\} \right) &\text{ if } D(a,b)  = 1\\
        \left(\argmax_{\succsim_1} \{ a' : (a',b) \in C \} ,b\right) &\text{ if } D(a,b) = 2
    \end{cases}
\end{equation}
\end{definition}

We first demonstrate that local dictatorship is necessary for individual strategy-proofness and Pareto efficiency with two agents.
 
 \begin{lemma}\label{local dictatorship 1}
  Suppose that $\vert N \vert =2$. For any constraint $C$, if $f:\pp\rightarrow C$ is strategy-proof and Pareto efficient, then it is a local dictatorship.
 \end{lemma}
 
 \begin{proof}
 
If $\bar{C}$ is empty, then all allocations are feasible and the unique Pareto efficient mechanism gives each agent her top choice at every profile. So assume $\bar{C}$ is nonempty and fix a strategy-proof and Pareto efficient mechanism $f:P^{2}\rightarrow C$.\footnote{Serial dictatorship always is Pareto efficient and strategy-proof, so one exists.} Let $a$ and $b$ be individually-feasible objects for $1$ and $2$ respectively. If $(a,b)$ is feasible, by Pareto efficiency, the allocation $(a,b)$ must be chosen. If $(a,b)\in \bar{C}$, since $a$ and $b$ are individually-feasible, there are $a'$ and $b'$ with $(a',b)$ and $(a,b')$ in $C$. Let $\succsim_{1}\in P^{\uparrow}\left[a,a'\right]$ and $\succsim_{2}\in P^{\uparrow}\left[b,b'\right]$. By Pareto efficiency, $f(\succsim_{1},\succsim_{2})=(a,b')$ or $f(\succsim_{1},\succsim_{2})=(a',b)$. Assume without loss that $f(\succsim_{1},\succsim_{2})=(a,b')$. We will show that this implies that $1$ is the local dictator at $(a,b)$. That is, for any other preference profile where $1$ top-ranks $a$ and $2$ top-ranks $b$, $1$ gets $a$ while $2$ gets her favorite object compatible with $a$. Pick any other $\succsim_{2}'$ which top-ranks $b$. By 2's strategy-proofness, $f_{2}(\succsim_{1},\succsim_{2}')\neq b$, but then from Pareto efficiency, $f_{1}(\succsim_{1},\succsim_{2}')= a$, since otherwise, the allocation $(a',b)$ would Pareto dominate $f(\succsim_{1},\succsim_{2}')$. Thus $f_{1}(\succsim_{1},\succsim_{2}')=a$ whenever $\succsim_{2}'\in P^{\uparrow}\left[b\right]$. Then using 1's strategy-proofness, we have that $f_{1}(\succsim_{1}',\succsim_{2}')=a$ for all $\succsim_{1}',\succsim_{2}'$ with $\tau(\succsim_{1}',\succsim_{2}')=(a,b)$. Finally, by Pareto efficiency, $f(\succsim_{1}',\succsim_{2}')=(a,\argmax_{\succsim_{2}'}\{ b' : (a,b') \in C\})$ whenever $\tau_{1}(\succsim_{1}',\succsim_{2}')=(a,b)$. Thus we say that $1$ is the local dictator at $(a,b)$. Since $(a,b)$, was arbitrary every other infeasible allocation in $\bar{C}^{*}$ has a local dictator by a symmetric argument.
\end{proof}

However, this only establishes necessity. Not all dictatorship partitions will be strategy-proof. For example, suppose agent $1$ is the local dictator when $(a,b)$ are the top choices, while agent $2$ is the local dictator at $(a,b')$. Then there are preference profiles where agent $2$ benefits from gaining local dictatorship rights by misreporting her top choice as $b'$. For this reason, we show that if $1$ is the local dictator at $(a,b)$, then individual strategy-proofness implies that she is also the dictator at $(a,b')$. Similarly, if $(a',b')$ is infeasible and $1$ is the dictator at $(a,b')$ while $2$ is the dictator at $(a',b')$, there are preference profiles where agent $1$ benefits from misreporting her top choice as $a'$, again to gain dictatorship rights. So strategy-proofness requires that the same dictator have control at any two infeasible joint allocations that agree on one of the dimensions of the allocation. A local dictator assignment $D:\bar{C}^{*} \rightarrow \{1,2\}$ is said to be \textbf{adapted} if for all $(a,b)$ and $(a',b')$ in $\bar{C}^*$,
\begin{align*}
D(a,b) = D(a',b') \text{ whenever } a=a' \text{ or } b = b'.
\end{align*}
and the resulting local dictatorship mechanism is also said to be adapted.

\begin{theorem} \label{two agent}
For any two-agent constraint $C$, a mechanism $f$ is strategy-proof and Pareto efficient if and only if it is an adapted local dictatorship. 
\end{theorem}

\begin{proof}
First we verify necessity of the axioms. The necessity of local dictatorship was established in Lemma \ref{local dictatorship 1}. It remains to show that the local dictatorship must be adapted. Suppose that $(a,b)$ and $(a',b')$ are two infeasible allocations in $\bar{C}^{*}$ and either $a=a'$ or $b=b'$. Without loss, assume $a=a'$. Suppose by way of contradiction that $(a,b)$ and $(a,b')$ have different local dictators. For example, suppose $D(a,b) = 1$ and $D(a,b') = 2$. Consider the preference profile $(\succsim_{1},\succsim_{2})$ where $\succsim_{1}\in P^{\uparrow}\left[a\right]$ and $\succsim_{2}\in P^{\uparrow}\left[b,b',b''\right]$ where $b''$ is such that $(a,b'')\in C$ (such a $b''$ exists because $a$ is individually-feasible). Then since $1$ is the local dictator at $(a,b)$, we get $f(\succsim_{1},\succsim_{2})=(a,b'')$. However, if $\succsim_{2}'\in P^{\uparrow}\left[b'\right]$, then $f_{2}(\succsim_{1},\succsim_{2}')=b'\succ_{2}b''=f_{2}(\succsim_{1},\succsim_{2})$ since $2$ is the local dictator at $(a,b')$. This gives a violation of strategy-proofness. Thus either $1$ is the local dictator at both $(a,b)$ and $(a,b')$ or $2$ is.

Now we turn to sufficiency. Suppose $f$ is an adapted local dictatorship. Efficiency is immediate. To verify strategy-proofness, fix a preference profile $(\succsim_1, \succsim_2)$. If agent $1$ is the dictator at $\tau_{1}(\succsim_1, \succsim_2)$, then she clearly has no reason to misreport. So suppose $1$ is not the dictator and let $\succsim'_1$ be an alternative preference for agent $1$. Either $\tau_{1}(\succsim'_1, \succsim_2)$ is feasible, or by adaptedness, agent $2$ is still the dictator at this profile because her top choice, say $b$, is the same as in $(\succsim_1, \succsim_2)$. Since $f_2 (\succsim'_1, \succsim_2) = b$, it must be that $f_1 (\succsim'_1, \succsim_2) \in \{a' : (a',b) \in C\}$ by feasibility. But then $f_1 (\succsim_1, \succsim_2) = \argmax_{\succsim_1} \{ a' : (a',b) \in C \} \succsim f_1 (\succsim'_1, \succsim_2)$, verifying strategy-proofness.
\end{proof}

One can easily establish the range of possibilities allowed by Theorem \ref{two agent}. Connect any two infeasible allocation $(a,b)$ and $(a',b')$ in $\bar{C}^*$ with an edge if either $a = a'$ or $b = b'$. This induces a graph $\Gamma(C)$ and adaptedness is simply the requirement that connected components of this graph share the same dictator. This process  is depicted in Figure \ref{two-agent example}, where we first remove objects that can't be allocated, then decompose the infeasible set into connected components.



\begin{corollary}
The number of strategy-proof and Pareto efficient mechanisms is $2^{Y}$ where $Y$ is the number of connected components of $\Gamma(C)$.
\end{corollary}

\begin{figure}[h]
\centering
\includegraphics[scale=0.56]{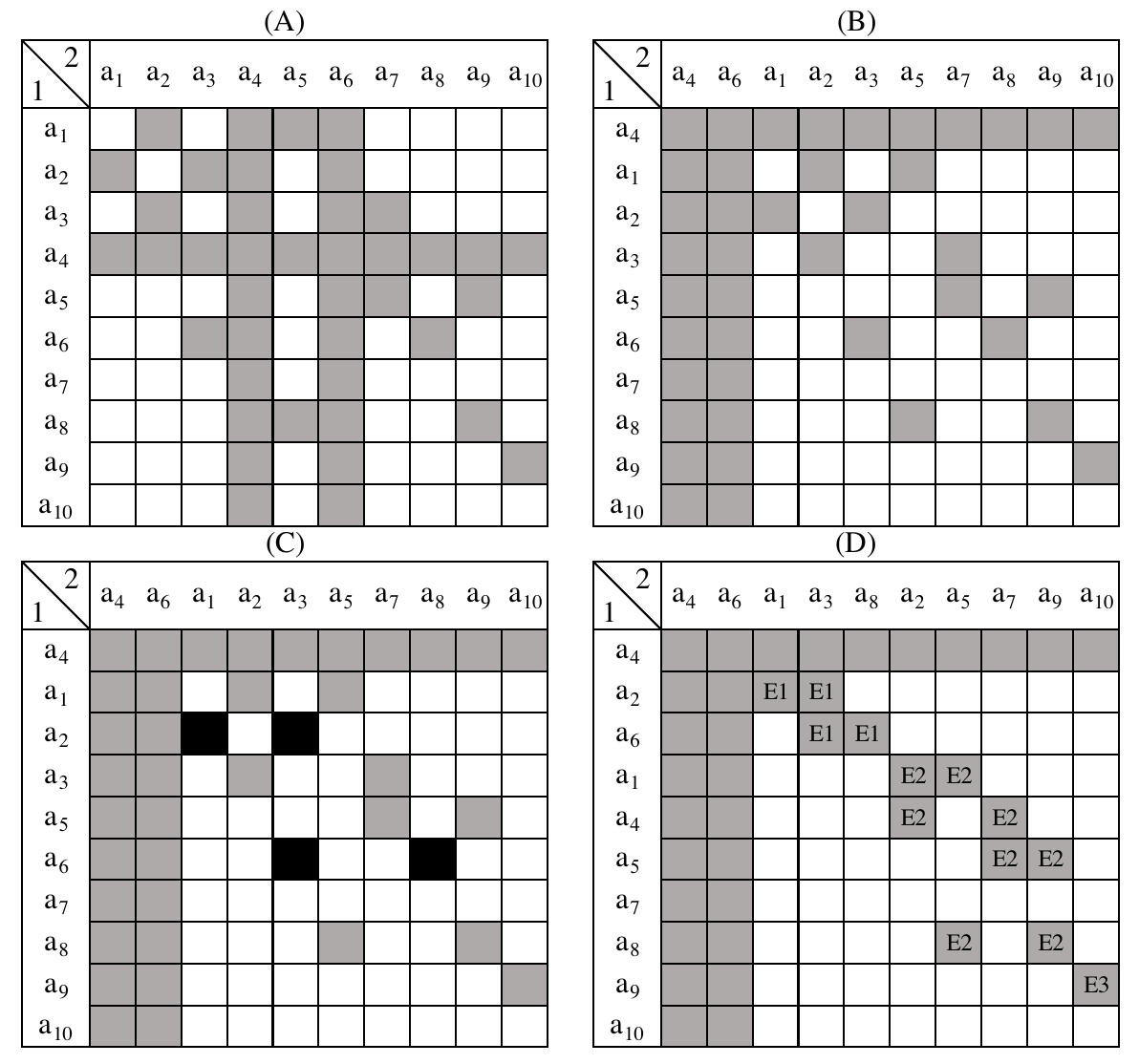}
\caption{Two-agent Example}
\label{two-agent example}
\end{figure}

In independent and contemporaneous work, \citeasnoun{Meng19} provides a characterization of all strategy-proof and Pareto efficient mechanisms for the two-agent social choice problem with exogenous indifferences. His characterization involves assigning a dictator at all profiles of preferences over announced indifference classes, where the dictator assignment must respect a cell-connected property. The structure of his characterization has similarities to our assignment of local dictators to the infeasible set. In fact, either result can be deduced from the other. While technically equivalent, these results are cast for very different questions, his for known indifference classes and ours for a known constraint. While the papers should share precedence for the mathematical result, the interpretations, and applications of the two papers differ.

\subsection{Many Agents}\label{sec: N agent}

We now turn to the problem of describing the class of group strategy-proof and efficient mechanisms for the many agent case. Before proceeding to the general problem, we start with a particular family of constraints where the analysis from Section \ref{two agent section} carries over with only slight modification. We call these constraints \textit{unilateral}. Unilateral constraints are of independent interest and contain task allocation problems and certain scheduling problems as special cases. 
We then turn to general constraints and show that a mechanism is group strategy-proof and efficient if and only if it is constructed by gluing together two-agent local dictatorship mechanisms. This characterization is less descriptive than the result given in Theorem \ref{two agent}. Nevertheless, in subsequent sections, we apply this characterization to many prominent constraints. 

\subsubsection{Unilateral Constraints}

Unilateral constraints are those constraints where, for any infeasible allocation, any agent can restore feasibility by changing their allocation alone.

\begin{definition}
A constraint $C$ is called \textbf{unilateral} if, for every infeasible $(a_{i})_{i\in N}$ and for every $i$, there is an object $a_{i}'$ such that $(a_{i}',a_{-i})$ is feasible. 
\end{definition}

We give several examples of unilateral constraints and then provide a characterization.

\begin{example}[Task allocation]\label{ex: task allocation}
    Consider a manager who needs to allocate a finite set $T$ of tasks to their employees. Objects $\obj = 2^T$ are subsets of tasks. The constraint is that every task be assigned to at least one worker, though employees can work together on a task and each worker can be assigned to multiple tasks. Therefore the constraint is $\{(a_i)_{i\in N} \sst \cup a_i=T\}$. In Figure \ref{fig: single compromising}, we show the constraint with three workers indexed by $1,2$ and $3$ and two tasks indexed by $a$ and $b$. Workers can be assigned no task -- $\emptyset$, task $a$ alone, task $b$ alone or both tasks $a$ and $b$. Note that this is a unilateral constraint since each worker can switch to take on all tasks and the outcome will be feasible, regardless of other workers' allocations. 

    \begin{figure}[h]
    \begin{center}
        \includegraphics[scale=0.58]{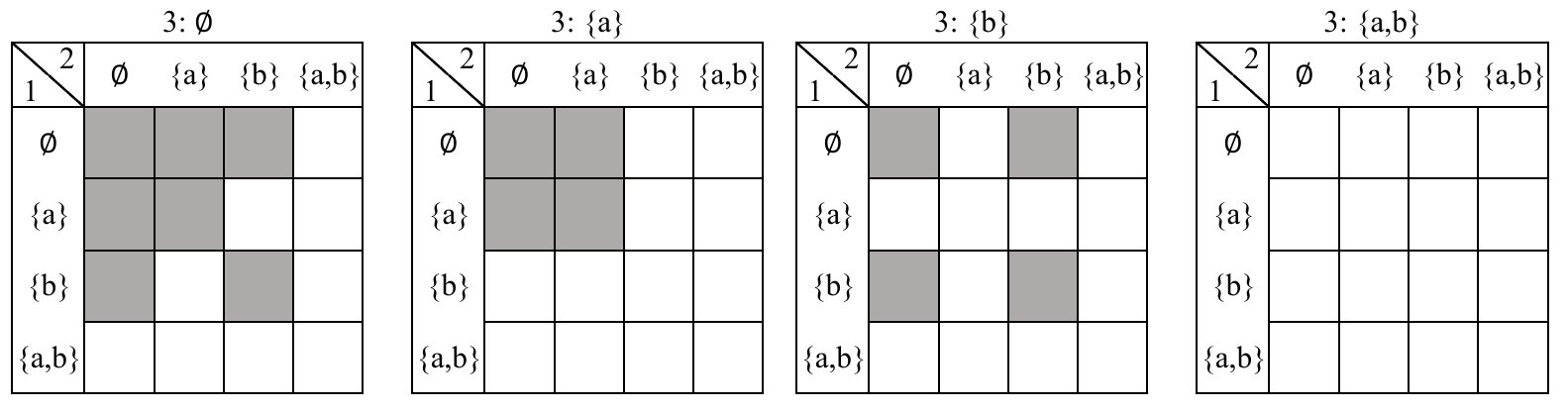}
        \label{fig: single compromising}
        \end{center}
        \caption{Task allocation constraint for three agents and two tasks. Agent $3$'s allocation is defined by the panel while agent $1$ and $2$'s allocation is determined by the row and column respectively. The infeasible allocations are shown in gray.}
    \end{figure}
\end{example}

\begin{example}[Airline regulation] \label{ex: airlines}
    Suppose that a regulator wants to ensure that airlines supply enough routes so that between every major city there is a way to get from one city to another, possibly with up to $k$ layovers. The airlines have preferences over their assigned bundles of routes. In this case, the constraint is unilateral in the sense that, for any network defined by some proposed routes, the regulator could force a single airline to supply enough missing routes to satisfy the connectedness requirements.\footnote{Before 1978 the Civil Aeronautics Board allocated routes to airlines in a centralized manner. Now, airlines decide their own route networks by bidding for gate rights at airports.}
\end{example}

Both examples are special cases of a more general construction. Whenever the objects to be allocated are themselves subsets of a set of basic items (e.g., sets of tasks and sets of routes in the two examples respectively) and the constraint depends only on the union and is monotone in the sense that every superset of a feasible union is feasible, then the constraint is unilateral. Formally, if $H$ is a basic collection of objects and $\obj=2^H$, then for a nonempty collection of subsets $\mathcal{H} \subseteq 2^H$ define $C_\mathcal{H} = \{(a_i)_{i \in N} : \bigcup_i a_i \in \mathcal{H} \}$. If $\mathcal{H}$ is closed under supersets, that is, $S' \in \mathcal{H}$ whenever $S \in H$ and $S' \supseteq S$, then $C_{\mathcal{H}}$ is unilateral. For example, the collection of $k$-path-connected route networks defines such a constraint on the allocations of routes, so the airline regulation example is unilateral. 

Another source of unilateral constraints are settings where a small portion of the allocations are infeasible. For example, suppose that each $a\in \obj$ is associated with a cost $p_a$ for the planner. The planner has a budget $b$ and $(a_i)_{i\in N}$ is feasible if and only if $\sum p_{a_i}\leq b$. As the budget $b$ grows, more allocations become feasible. As soon as $b\geq (\vert N \vert -1)p_{\max} +p_{\min}$, where $p_{\max}$ and $p_{\min}$ are the most and least costly object prices respectively, the constraint is unilateral.

With unilateral constraints, every group strategy-proof and Pareto efficient mechanism can be described in a simple manner analogous to the characterization of the two-agent case. In fact, this is what makes the two-agent case special: up to removing objects that can never be allocated (along the lines of Lemma \ref{dumb objects dont matter}), every two-agent constraint is unilateral. So, unilateral constraints are theoretically helpful to understand the scope and limits of the two-agent result because these are cases where that simple characterization extends to many agents.

For any unilateral constraint $C$, a \textbf{local compromiser assignment} is a function $r$ that associates an agent to every infeasible allocation $\bar{C}$. We require that if $r(a)=i$ for some infeasible allocation, then for any other infeasible allocation  $(a_{i}',a_{-i})$, we must also have $r(a_{i}',a_{-i})=i$. A local compromiser assignment $r$ induces a \textbf{unilateral mechanism} that works as follows. For any preference profile $\succsim$, give every agent her top choice, resulting in the allocation $a$, if it is feasible to do so. Otherwise, when $a\in \bar{C}$, give the agent $r(a)$ their favorite alternative compatible with $a_{-r(a)}$ and give all other agents $j$ their favorite objects $a_j$. By definition, such an alternative exists and the associated allocation will be feasible. The following proposition shows that unilateral mechanisms are group strategy-proof and efficient and that all group strategy-proof and efficient mechanisms are unilateral mechanisms. 

\begin{theorem} \label{single-compromising}
Let $C$ be unilateral. A mechanism $g$ is group strategy-proof and Pareto efficient if and only if it is a unilateral mechanism.
\end{theorem}

Returning to the task allocation example, Theorem \ref{single-compromising} can be used to generate group strategy-proof and Pareto efficient mechanisms. 

\paragraph{Example \ref{ex: task allocation} Continued.} Figure 4 gives an example of a function $r$ which induces unilateral mechanism for the constraint in Example \ref{ex: task allocation}.
    \begin{figure}[h]
    \begin{center}
        \includegraphics[scale=0.58]{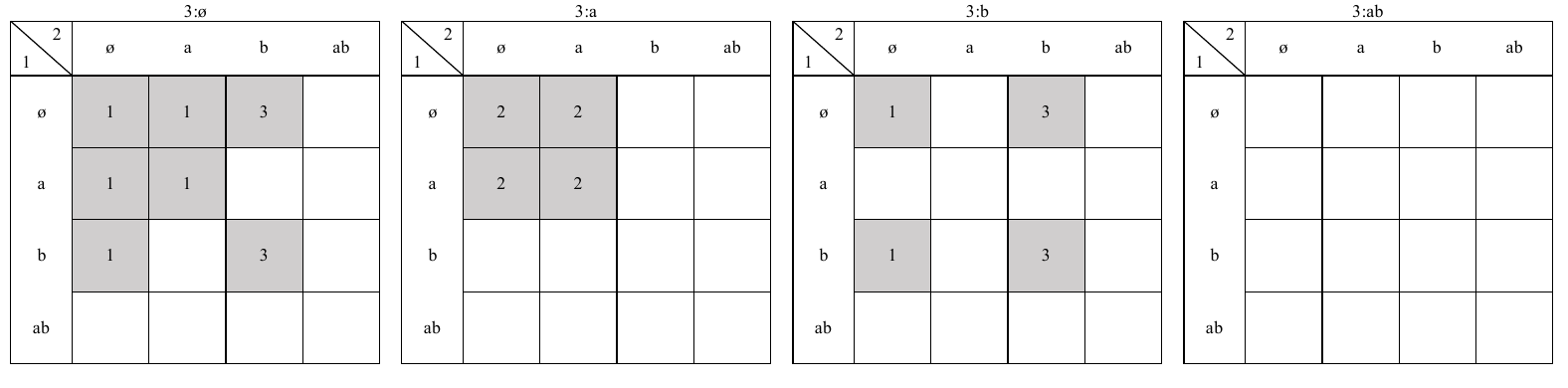}
        \label{fig: single compromising 2}
        \end{center}
        \caption{An example of a unilateral mechanism for a three agents and two task allocation problem.}
    \end{figure}
    Notice that this mechanism is non-dictatorial and no worker is guaranteed her top choice. A single worker is required to take-on any remaining tasks from an infeasible allocation. However, this worker is determined endogenously. For example, if $1$ top-ranks $\emptyset$ and $2$ top-ranks $b$ then $3$ must chose between $\{a\}$ and $\{a,b\}$.

\subsubsection{A General Characterization}

We now turn to general constraints. It will be useful to first develop the notion of a marginal mechanism. 

\begin{definition}
Let $f:\pp\rightarrow C$ and let $M$ be a proper subset of $N$. Let $\succsim_{-M}$ be a profile of preferences for the agents not in $M$. The \textbf{marginal mechanism} at $\succsim_{-M}$ is denoted $f_{\succsim_{-M}} :P^{M} \rightarrow \obj^{M}$ and defined as the function $$\succsim_M \enspace \mapsto \, \left[f_{j}(\succsim_M, \succsim_{-M}) \right]_{j\in M}.$$ We denote $I(\succsim_{-M})=f_{\succsim_{-M}}(P^M)$ which will be referred to as $M$'s \textbf{marginal option set}.
\end{definition}

Thus a marginal mechanism holds fixed some of the agents' preferences $\succsim_{-M}$ and defines an $\vert M\vert$-agent mechanism for the remaining agents, mapping their profile of announcements $\succsim_{M}$ to an $M$-agent allocation. For strategy-proof mechanisms, the single-agent marginal mechanisms are especially simple. They simply ask the agent pick from a fixed menu, often called the ``option set." In fact, a mechanism is strategy-proof if and only if the single-agent marginal mechanisms are of this type. To our knowledge, \citeasnoun{Barbera83} was the first to make this observation. Single-agent marginal mechanisms have also been used in the experimental literature to simplify the description of mechanisms, for example by \citeasnoun{gonczarowski2023strategyproofness}. Clearly, marginal mechanisms inherit the group strategy-proofness of the original mechanism. It turns out that in the other direction, group strategy-proofness of the two-agent marginal mechanisms suffices for groups strategy-proofness of the full mechanism. 

\begin{lemma}[\citeasnoun {Alva17}, Theorem 1] \label{N agent}
The mechanism $f:\pp\rightarrow C$ is group strategy-proof if and only if it is pairwise strategy-proof. 
\end{lemma}

This significantly reduces the number of conditions one need to check to verify that a mechanism is group strategy-proof. Rather than verifying incentives for all coalitions, it is sufficient to check that no two agents can profitably misreport their preferences. The equivalence of group and pairwise incentives was originally proved for the house allocation domain by \citeasnoun{Papai00}. The most general equivalence was proven by \citeasnoun{Alva17} for a broad set of preference domains. We give a more direct proof for our setting in the appendix. The main intuition comes from the fact that group strategy-proofness is equivalent to individual strategy-proofness plus a condition called non-bossiness. Non-bossiness requires that no agent can affect the outcome of another agent by changing her report, unless it also changes her allocation. Non-bossiness is a pairwise condition and therefore group strategy-proofness can be checked by consulting only single agents and pairs. 

While this result takes a significant step towards understanding group strategy-proofness, it is especially useful in light of our characterization of two agent strategy-proof and efficient mechanisms. For two-agent mechanisms, there is only one group coalition, namely the grand coalition. Therefore, group strategy-proofness of a two-agent mechanism is equivalent to individual strategy-proofness and Pareto efficiency on its image. Combining Lemma \ref{N agent} with Theorem \ref{two agent}, we get a description of all group strategy-proof and Pareto efficient mechanisms. 

\begin{theorem}\label{N agent combined}
A mechanism $f:\pp\rightarrow C$ is group strategy-proof and Pareto efficient if and only if $f_{\succsim_{-ij}}$ is an adapted local dictatorship (using the marginal constraint) for any two agents $i,j$ and any residual preference profile $\succsim_{-ij}$.
\end{theorem}

Compared to Theorem \ref{two agent}, this characterization is considerably less descriptive. Theorem \ref{two agent} gives an explicit and simple description that captures all strategy-proof and Pareto efficient mechanisms for an arbitrary two agent constraint.  Theorem \ref{N agent combined} describes such mechanisms more indirectly. Nevertheless, as we will see in our applications, this characterization substantially reduces the burdens in verifying and constructing group strategy-proof and Pareto efficient mechanisms for many constraints.

\section{Existence Results}\label{sec: existence}

In this section, we define sequential dictatorships, a class of mechanisms that generalize the familiar serial dictatorship mechanisms and that are group strategy-proof and Pareto efficient for any constraint. We then provide a theorem that provides conditions on the constraint $C$ that are sufficient to guarantee that $C$ has group strategy-proof and efficient mechanisms that are not sequential dictatorships. As a corollary, we deduce that the two-sided matching and school choice problems admit non-dictatorial mechanisms. 

\subsection{Sequential Dictatorships}

We begin by formalizing the sequential dictatorship mechanism. This mechanism is always group strategy-proof and Pareto efficient. In the more familiar serial dictatorship, agents are ordered exogenously, and each agent in turn selects her most preferred object that is feasible given the choices of earlier dictators. Sequential dictatorship generalizes this idea by allowing the order in which remaining agents are called to depend on the choices made by earlier dictators.\footnote{To our knowledge, sequential dictatorships were first described by \cite{Papai01} in the multiple-assignment problem, with related formulations appearing in \citeasnoun{KlausMiyagawa02}, \citeasnoun{EhlersKlaus03}, and \citeasnoun{Hatfield09}. These papers show that sequential dictatorship is the unique mechanism satisfying group strategy-proofness and Pareto efficiency in various multi-object assignment environments. \citeasnoun{PyUn25} explore sequential dictatorships in a general model and establish cases where ordinality, group strategy-proofness and efficiency imply sequential dictatorship. }

A shortcoming of sequential dictatorship is that it concentrates power to the early dictators, so is sometimes considered an unfair mechanism. However, in some situations sequential dictatorship is the only available group strategy-proof and Pareto efficient mechanism. The applications in the sequel of this section demonstrate that characterization for social choice, one-sided matching, and multiple assignment. 

We now formally define sequential dictatorship. Recall that $\salloc$ is the set of suballocations (i.e. the maps $\mu:M\rightarrow \obj$ where $M\subset N$). Let $\salloc'$ be the set of incomplete suballocations.\footnote{Where $M$ is a proper subset of $N$.} A \textbf{picking order} is a map $\zeta: \salloc'\rightarrow N$ such that for any suballocation $\mu$, $\zeta(\mu)$ is an agent not allocated an object under $\mu$. For each picking order and for any constraint $C$ we may define the \textbf{sequential dictatorship} for $\zeta$ to be the mechanism whose allocation at any preference profile is determined by the following algorithm:
\vspace{0.2cm}
\begin{tcolorbox}
\fbox{Step 1} \hspace{0.3cm} The agent $d_{1}\equiv \zeta(\emptyset)$ is the first dictator. She chooses her favorite object $a_{1}$ from $\pi_{d_{1}}C$. Let $\mu_{1}$ be the suballocation where $d_{1}$ is assigned $a_{1}$ and all other agents are unassigned.
\vspace{0.1cm}

\noindent \fbox{Step k} \hspace{0.3cm} The agent $d_{k}\equiv \zeta(\mu_{k-1})$ chooses her favorite object $a_{k}$ from $\pi_{d_{k}}C(\mu_{k-1})$. Let $\mu_{k}$ be the allocation that agrees with $\mu_{k_1}$ and additionally assigns $d_k$ the object $a_k$. If all agents have been assigned an object, stop. If not, continue to step $k+1$.
\end{tcolorbox}

Standard serial dictatorship, where the order of dictators is fixed and invariant to earlier choices, are a prominent special case of sequential dictatorships where the picking order only depends on the agents identified in a suballocation, that is where $\zeta(\mu) = \zeta(\mu')$ whenever $\mu$ and $\mu'$ are suballocations to the same subcoalition $M$ of agents.

Notice that a single mechanism can be the result of many picking orders. This is because the picking order $\zeta$ can be defined in any way off the ``algorithm path,'' in the sense that suballocations which will never be realized can be assigned any agent as the next dictator. For example, in the serial dictatorship mechanism, any allocation in which a single agent other than the dictator is assigned an object will never be realized, so the picking order there is immaterial. Keeping this redundancy in mind, it is nonetheless convenient to take $\salloc'$ as the domain of all picking orders. The following remark implies non-emptiness of the set of group strategy-proof and Pareto efficient mechanisms. 
\begin{remark}\label{generalized_serial_dictatorships}
For any constraint $C$, any sequential dictatorship is group strategy-proof and Pareto efficient.
\end{remark}
\noindent These properties of sequential dictatorship are extended from serial dictatorship, where the ordering is independent of the profile.

We say that a constraint $C$ admits \textbf{non-dictatorial} group strategy-proof and efficient mechanism if there are group strategy-proof and efficient mechanisms which are not sequential dictatorships. We will discuss in the next section why some constraints like social choice, multiple assignment, and the roommates problem are dictatorial.  \citeasnoun{PyUn25} provide a more focused analysis of sequential dictatorships. They take ordinality, which includes group strategy-proofness and Pareto efficiency as cases, as their overarching criterion and unify different cases where sequential dictatorships are the only ordinal mechanisms. Their result is more general in that it applies to preference domains beside strict profiles. \citeasnoun{PyTr23} show that mechanisms that resemble sequential dictatorship are simple for agents to understand in the sense that incentives are very clearly strategy-proof.

\subsection{The Existence of Non-dictatorial Mechanisms}

We now give a partial converse to the Gibbard-Satterthwaite Theorem. That is, we provide conditions which guarantee that a constraint admits non-dictatorial group strategy-proof and efficient mechanisms.

Fix a constraint $C$. If there is a pair of agents $i,j$ such that $\Gamma(C^{\{i,j\}})$ has at least two connected components then by Theorem \ref{two agent} we can construct a non-dictatorial mechanism $f$ for $i$ and $j$ on the constraint $C^{\{i,j\}}$.\footnote{Recall that $\Gamma(C)$ is the graph formed by taking as nodes the set of infeasible allocations and including an edge between any two infeasible allocations that agree on one coordinate.} This mechanism can be extended to the remaining agents using sequential dictatorship. This is summarized by the following theorem.

\begin{theorem}\label{GS converse}
If a constraint $C$ is such that for some $i,j$, the graph $\Gamma(C^{i,j})$ has more than one component, then $GS(C)$ is strictly larger than the set of sequential dictatorship mechanisms. 
\end{theorem}

This result provides a simple test to verify that non-dictatorial mechanisms exist. Many known cases where other mechanisms are available fall under this result.

\begin{corollary}\label{cor: existence}
    The following settings admit non-dictatorial mechanisms:
    \begin{itemize}
        \item Two-sided matching with at least two agents on each side
        \item School choice with at least two schools $s$ and $t$ with capacity $q_s$ and $q_t$ such that $q_s+q_t \leq N $
        \item House allocation with at least two houses.
    \end{itemize}
\end{corollary}

As shown by the characterization for house allocation mechanisms due to \citeasnoun{PyUn17}, an exact description of the class of mechanisms can be difficult and subtle when there are non-dictatorial mechanisms. \citeasnoun{root2023royal} characterize the group strategy-proof and efficient two-sided matching mechanisms under an additional fairness requirement that requires the mechanism to treat the two sides symmetrically. To the best of our knowledge, absent such a condition, no characterization is known. The school choice setting is another important outstanding problem. Top trading cycles can be adapted to school choice \cite{abdulkadirouglu2003school}. While deferred acceptance is not group strategy-proof or efficient in general, under an acyclicity condition on priorities, it recovers these properties \cite{ergin2002efficient}. Any characterization will have to nest these two classes of mechanisms. 

While we are primarily interested in applying the two-agent characterization to general constraints in this paper, in a companion paper, we provide a technique for constructing new mechanisms for arbitrary constraints \cite{root2023local}. We call these ``local priority mechanisms." There is no a priori guarantee that local priority mechanisms will capture all group strategy-proof and efficient mechanisms. Nevertheless, a wide variety of mechanisms fit into this category including deferred acceptance \cite{GaSh62} and trading cycles mechanisms \cite{PyUn17}. 
 
\section{Characterizations for Specific Constraints} \label{sec: specific constraints}

We now turn to a number of specific families of constraints. In each case, we apply Theorems \ref{two agent} and \ref{N agent combined} to recover the class of group strategy-proof and efficient mechanisms.

\subsection{Social Choice}

We first turn to the social choice problem and prove the celebrated Gibbard-Satterthwaite Theorem \cite{Gibbard73,Satterthwaite75}. The goal is to give a simple demonstration of how to apply our characterization results.

\begin{lemma} \label{GSP=ISP for social choice}
Let $C$ be the social choice constraint, i.e. $C=\{(a_{i})_{i\in N} \sst a_{i}=a_{j}\text{ for all }i,j\in N \}$ then a mechanism $f:\pp \rightarrow C$ is group strategy-proof if and only if it is individually strategy-proof.
\end{lemma}

We can then apply our main characterization results to the special case of the diagonal social choice constraint to derive that all group strategy-proof and onto mechanisms are dictatorships, which by virtue of Lemma \ref{GSP=ISP for social choice} is equivalent to the Gibbard--Satterthwaite Theorem.

\begin{theorem}[Gibbard--Satterthwaite] \label{GS Theorem}
Let $C$ be the social choice constraint. If $\vert \obj \vert >2$ and $f:\pp\rightarrow C$ is surjective and group strategy-proof then it is dictatorial.\footnote{In fact, we only need that $\vert im(f)\vert >2$ in which case we could drop items never allowed and recover the same statement.}
\end{theorem}

We leave the formal proof to the appendix but we sketch the argument here. First observe that some two-agent marginal mechanism has at least three outcomes. By way of contradiction, suppose all have at most two. Some marginal mechanism, say for agents $1$ and $2$ when others report $(\succsim_3, \ldots, \succsim_n)$, has exactly two outcomes since otherwise the whole mechanism is single-valued. A simple consequence of Theorem \ref{two agent} that we discussed earlier is that the marginal mechanism with two outcomes ($a$ and $b$) is either a dictatorship (say of agent $1$)  or a ``veto'' mechanism where one object (say $b$) is the default unless both prefer $a$. Suppose player $2$'s preference $\succsim^*_2$ has  $c \succ^*_2 a \succ^*_2 b$. Then agent $1$'s preference is followed in both dictatorship and the veto mechanism. But the marginal mechanism for agents $1$ and $3$ when others report $(\succsim^*_2, \succsim_4, \ldots, \succsim_n)$ must also have $a$ and $b$ as outcomes. Repeating the argument across agents, $a$ or $b$ is implemented even though all agents prefer $c$, violating group strategy-proofness.

Any two-agent mechanism (say for agents $1$ and $2$) with three or more outcomes is dictatorial (say for $1$) by visual inspection. Then $2$'s preference is irrelevant, so there also exists a $(1,3)$-marginal mechanism with at least three outcomes. Agent $1$ must still be dictator, so $3$'s preference is irrelevant. Repeating across agents, no preference matters besides $1$'s so she is a global dictator.

Interestingly, the group strategy-proof and Pareto efficient mechanisms for the roommates problem and the social choice problem are the same: sequential dictatorships. For the roommates problem, there is enough flexibility in the constraint that a sequential dictatorship still has room for the those who are second or later in the picking order to have nontrivial choices. In the social choice problem, all sequential dictatorships are (simple) dictatorships because agents picking second or later have no choices to make because the first dictator's choice determines the entire allocation profile. So our model presents a unified view of both problems as allowing only sequential dictatorships.

\subsection{The Roommates Problem}

We now apply our general results to the canonical roommates problem. 
In the roommates problem, an even number of agents need to be paired as roommates. Each agent has a strict preference over the other agents as roommates. We assume that all agents prefer to be matched rather than be unmatched, so we do not consider problems related to individual rationality. As discussed earlier, we can model this in our environment by letting $\obj=N$ and using the constraint $$C=\{\mu:N\rightarrow N \sst \mu(i)\neq i \text{ for all }i\text{ and }\mu\circ \mu=id\}$$ where $id$ is the identity ($id(i)=i$ for all $i$). Any feasible mechanism for this constraint will be called a \textbf{roommates mechanism}. As mentioned in the introduction, the literature on the roommates problem has focused on stable matching and there is little known about incentives and efficiency for this problem.

Theorem \ref{Roommates_Characterization} characterizes all group strategy-proof and Pareto efficient roommates mechanisms. This is akin to the Gibbard--Satterthwaite Theorem that demonstrates all such mechanisms are dictatorships for the social choice problem and related to the recent result of \citeasnoun{PyUn17} that characterizes all such mechanisms for the house allocation problem. We settle this question for the roommates problem, and show that all mechanisms with these properties are sequential dictatorships.

\begin{theorem} \label{Roommates_Characterization}
A roommates mechanism is group strategy-proof and Pareto efficient if and only if it is a sequential dictatorship.
\end{theorem}

While the proof of Theorem \ref{Roommates_Characterization} is involved and therefore only fully described in the Appendix, we briefly sketch some of the main ideas here. Suppose that $f$ is a group strategy-proof and Pareto efficient roommates mechanism. Consider any pair $i$ and $j$ and any preference profile $\succsim_{-ij}$. Let $X_{i}$ and $X_{j}$ denote the set of objects (partners) with whom $i$ and $j$ cannot be matched in the marginal mechanism $f_{\succsim_{-ij}}$. Lemma \ref{dumb objects dont matter} implies that $X_i$ and $X_j$ can be ignored in the  marginal mechanism. For this discussion, assume that it is possible for $i$ and $j$ to be matched with each other in this marginal mechanism. The collection of infeasible allocations in the marginal constraint is a superset of those shown in the left panel of figure \ref{roommates_1}, but by assumption does not contain the allocation $(j,i)$.

\begin{figure}[h]
    \centering
    \includegraphics[scale=0.6]{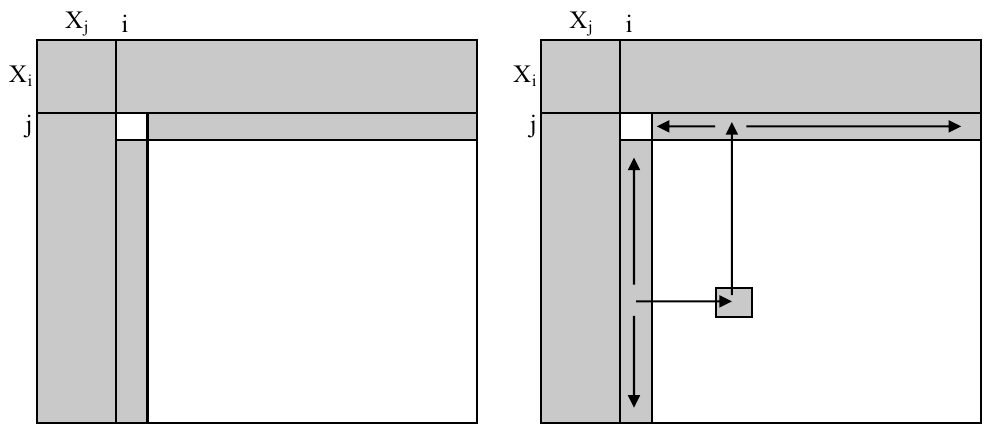}
    \caption{The marginal constraint $I (\succsim_{-ij})$}
    \label{roommates_1}
\end{figure}

Notice that if there are any infeasible allocations in the set $(\obj - X_{i}-\{j\})\times (\obj - X_{j}-\{i\})$ then the graph $\Gamma(I_{\succsim_{-ij}})$ of this marginal constraint is totally connected and there must be a single dictator that is common to all infeasible allocations, as illustrated in the right panel of figure \ref{roommates_1}. Otherwise, if the infeasible space is not connected, then the infeasible space is as illustrated in the left panel of figure \ref{roommates_1}. For that case, the picture strongly resembles the two-agent social choice problem with two objects. We discussed the possible mechanisms for that case after figure \ref{two-agent examples}, showing that these are either dictatorship or unanimity/veto mechanisms with a default option. From Theorem \ref{two agent}, we know that each of the two disconnected components must be assigned a dictator, leaving four possible mechanisms as illustrated in figure \ref{roommates_2}.

\begin{figure}[h]
    \centering
    \includegraphics[scale=0.6]{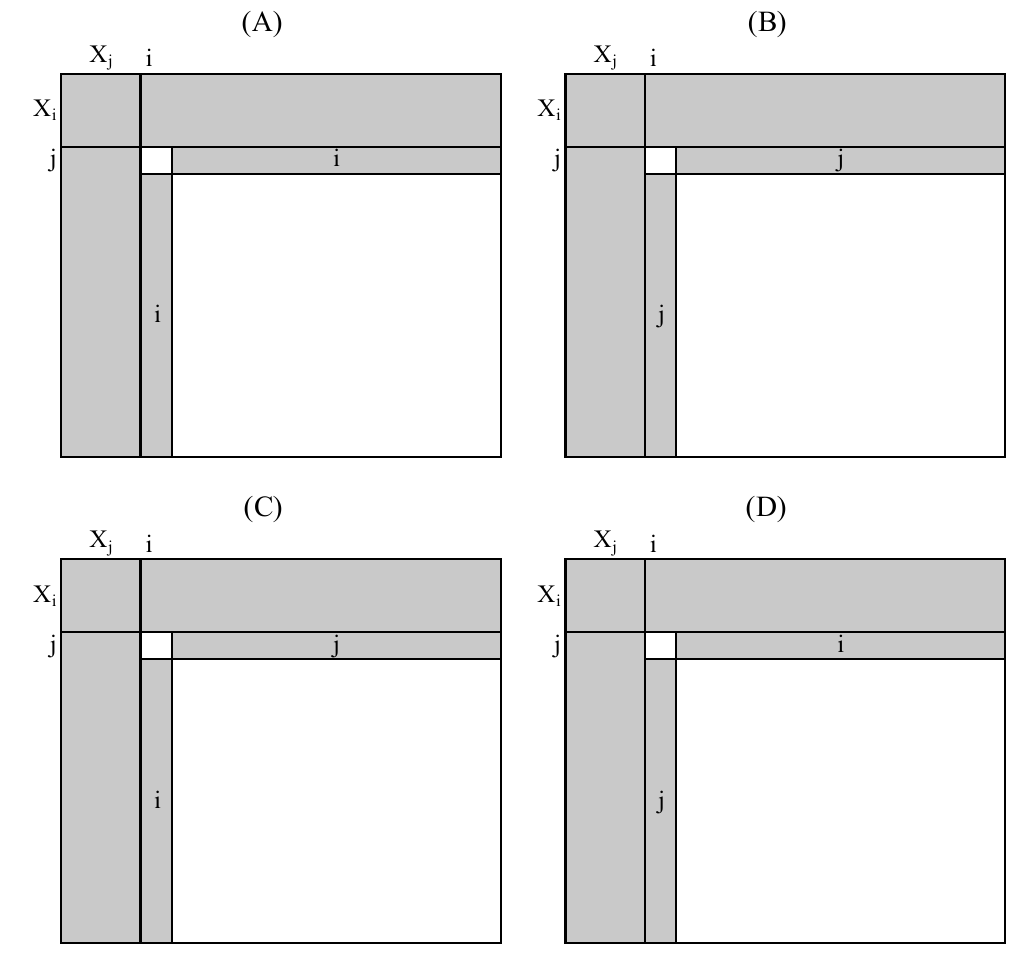}
    \caption{The possible marginal mechanisms $f_{\succsim_{-ij}}$}
    \label{roommates_2}
\end{figure}

Panels (A) and (B) assign the same dictator to both regions and are therefore standard dictatorships. Panels (C) and (D) are veto mechanisms. In (C), $i$ and $j$ are matched together if either top-ranks the other. This is like a veto mechanism with the default option being ``match together.'' In (D), they match  only if both $i$ and $j$ top-rank one another, which is like a veto option with the opposite default of ``do not match together.'' The proof proceeds by induction on the number of agents. The specific steps involve verifying that mechanisms (C) and (D) are not possible for the two-agent marginal mechanisms because they are precluded by the structure of the grand constraint across all agents in the roommates problem. While this still requires work, it greatly simplifies the problem since we are required only to rule out this specific type of submechanism. We mention this here as an example of how central the two-agent characterization can be for proving results with many agents.

Although the results in this paper are generally unrelated to stability, this one does speak to a stability. In fact, it immediately exposes a tension between incentives and stability. As mentioned, an important feature of the roommates problem is the lack of stable outcomes. A standard escape to find positive results is to demand something weaker than standard stability. For example, one relaxation only requires that pairs of agents where each ranks the other as her favorite must be matched, substantially reducing the set of relevant blocking pairs. This weaker stability condition is called ``mutually best'' by \citeasnoun{Toda06} and ``pairwise unanimity'' by \citeasnoun{tak10}. However, sequential dictatorships violate even this very weak form of stability. So a corollary of Theorem \ref{Roommates_Characterization} is that group strategy-proofness and Pareto efficiency are incompatible with even this weak stability notion, exposing a strong tension between incentives and stability for the roommates problem. This negative observation for the roommates problem is not new; in fact, this corollary of our result can also be implicitly derived from Theorem 2 of \citeasnoun{Takamiya13} without an explicit characterization of group strategy-proofness.\footnote{We thank Yuichiro Kamada for pointing this out to us.}   We mainly provide this result to show that our results regarding incentives and efficiency can indirectly provide insights into stability.
\begin{corollary}[\citeasnoun{Takamiya13}]
No group strategy-proof and Pareto efficient roommates mechanism can guarantee that mutual top choices are matched.
\end{corollary}

Finally, we note that here we do not allow for the case of self-assignment, where an agent can be assigned to herself, which can be interpreted as either exiting the market or being assigned a single room. In route to their findings for two-sided matching, \citeasnoun{root2023royal} examine this generalization of the one-sided matching problem.

\subsection{Multiple Assignment}

Consider the problem where there is a finite set $H$ of ``houses" and each agent can own up to $k$ houses. In this case $\obj$ is the set of subsets of $H$ of size $k$ or less. We allow agents to have arbitrary preferences over these subsets. An allocation $(s_i)_{i\in N}$ is feasible if for every $i$ and $j$, $s_i\cap s_j = \emptyset$. The house allocation problem is the special case where $k=1$.  In the special case where $k=2$, there are four houses $H=\{a,b,c,d\}$ and there are two agents $1$ and $2$, the constraint is shown in Figure \ref{fig: combinatorial assignment}. 
\begin{figure}[h]
    \centering
    \includegraphics[scale=0.6]{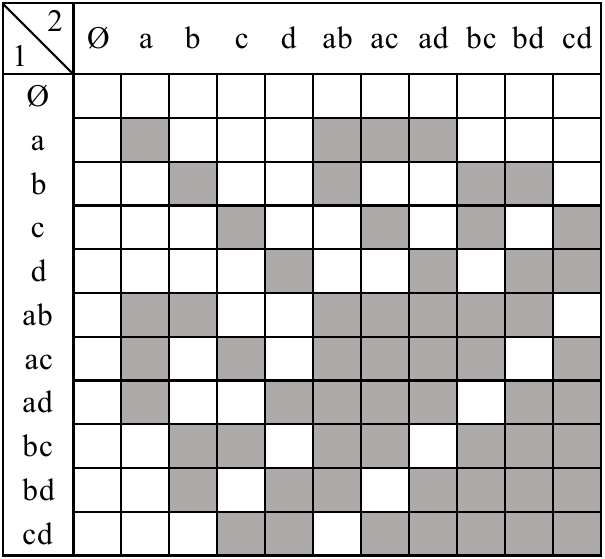}
    \caption{The constraint for the multiple assignment problem for two agents where each can each be assigned up to two objects from the set $\{a,b,c,d\}$.}
    \label{fig: combinatorial assignment}
\end{figure}
Applying Theorem \ref{two agent} to this particular constraint, simple visual inspection verifies that the entire infeasible space is connected, so the only strategy-proof and efficient mechanisms are sequential dictatorships, recalling that dictatorship and sequential dictatorship are equivalent with two agents. This turns out to be generally true. We use our prior characterizations to prove following characterization for any number of agents and any $k\geq 2$.

\begin{theorem}
If $k\geq 2$ all group strategy-proof and Pareto efficient mechanisms are sequential dictatorships. 
\end{theorem}

\begin{proof}
We will prove the result by induction over the number of agents. Our base case begins with just two agents $1$ and $2$. If $\vert H \vert =1 $, there is a single infeasible allocation, so the result holds immediately by Theorem \ref{two agent}. Now suppose that $\vert H \vert \geq 2$ and that $(s_1,s_2)$ and $(t_1,t_2)$ are infeasible. We need to show that the same local dictator will be assigned to both $(s_1,s_2)$ and $(t_1,t_2)$. Let $x$ be a house in $s_1\cap s_2$ and $y$ be a house in $t_1\cap t_2$ (we don't rule out the case that $x=y$). Then every allocation in the sequence $(s_1,s_2),(s_1,\{x,y\}),(\{x,y\},\{x,y\}),(\{x,y\},t_2),(t_1,t_2)$ is infeasible (If $x=y$, the set $\{x,y\}=\{x\}=\{y\}$). Furthermore, each infeasible allocation differs from the previous one by a single agent's allocation. In light of Theorem \ref{two agent}, any Pareto efficient and strategy-proof mechanism must have a single dictator assigned to the entire constraint. 

We proceed by induction on the number of agents. Suppose that the theorem holds for $n$ agents. Fix a group strategy-proof and efficient mechanism $f$ for the $n+1$ agents $\{0,1,2,\dots, n\}$. It will be enough to show that a single agent always gets their top choice since after this agent gets their top choice, we have a sub-problem for the remaining agents which is covered by the induction assumption. We establish that at least one agent must get their top choice with a series of facts, each of which rely on the induction assumption:

\textbf{Fact 0:} For any preference profile $\succsim$ where some agent $i$ top-ranks $\emptyset$, $i$ is matched with $\emptyset$  by $f$ and all other agents are matched via some sequential dictatorship. The picking order is independent of $i$'s ranking.

\textit{Proof of fact 0:} $i$ must be matched with $\emptyset$ by Pareto efficiency. If not, there is a Pareto improvement where $i$ is matched with $\emptyset$ and all other agents' allocations are unchanged. The marginal mechanism for agents other than $i$ is now a sequential dictatorship by the induction assumption. By group strategy-proofness the picking order is independent of $i$'s preference. That is, the picking order does not depend on $i$'s ranking of alternatives below $\emptyset$. 

\textbf{Fact 1} There is an agent $i$ such that for any $x\in \obj$ and any preference profile where all agents top-rank $x$ and second-rank $\emptyset$, $i$ gets $x$.

\textit{Proof of Fact 1:} By Pareto efficiency at every preference profile $\succsim$ described above exactly one agent gets assigned $x$ and the other agents are assigned $\emptyset$. Suppose that $\succsim$ is a profile where all agents top-rank $x$ and second-rank $\emptyset$ and $\succsim'$ is a profile where all agents top-rank $y$ and second-rank $\emptyset$. Since there are at least three agents, at least one agent $j$ is assigned $\emptyset$ at both $\succsim$ and $\succsim'$. However, by group strategy-proofness, if $j$ pushes $\emptyset$ to the top of their ranking, no one's assignment should change. By Fact $0$, this implies that there is an agent $i$ who gets $x$ at $\succsim$ and $y$ at $\succsim'$.

\textbf{Fact 2}
Let $i$ be the agent guaranteed to exist by Fact $1$. For any preference profile $\succsim$ where at least one agent gets $\emptyset$, $i$ must get their top choice. 

\textit{Proof of Fact 2:}
Let $j$ be the agent who gets $\emptyset$. By group strategy-proofness, the outcome doesn't change if $\succsim_{j}$ is changed to a profile which top-ranks $\emptyset$. However in this case, the marginal mechanism for the agents other than $j$ must be a sequential dictatorship by the induction assumption. The only agent who can be the first dictator in this mechanism is agent $i$ by Fact $1$. 

\textbf{Fact 3}
Let $i$ be the agent guaranteed to exist by Fact $2$. Let $j$ be the agent who is the first dictator in the marginal mechanism when $i$ top-ranks $\emptyset$. This agent exists by Fact $0$. For any profile $\succsim_{-ij}$ the marginal mechanism for $i$ and $j$ is a dictatorship. 

\textit{Proof of Fact 3:}
Let $g$ be the marginal mechanism for $i$ and $j$ at $\succsim_{-ij}$. Now consider the marginal option set $I(\succsim_{-ij})$. By definition, $I(\succsim_{-ij})$. Note that for any $x\in \obj$ if $i$ top-ranks $\emptyset$ and $j$ top-ranks $x$ then $j$ gets $x$. Likewise if $i$ top-ranks $x$ and $j$ top-ranks $\emptyset$ then $i$ gets $x$ by Fact 2. This implies that all outcomes are individually-feasible for both agents. Take any $(x,y)$ and $(x',y')$ which are both outside of $I(\succsim_{-ij})$ (so that both are infeasible for $i$ and $j$ in the marginal mechanism $g$). Now $(y,y)$ and $(y',y')$ are both in the complement of $C^{ij}$ so are both in the complement of $I(\succsim_{-ij})$. Let $\Gamma$ be the graph whose vertices are the elements of $\bar{I(\succsim_{-ij})}$ and where any two vertices $(u,v)$ and $(u',v')$ are connected if $u=u'$ or $v=v'$. We have already established that $(y,y)$ and $(y',y')$ are in the same connected component of $\Gamma$ in the discussion of the base case. Furthermore, $(x,y)$ is connected to $(y,y)$ and $(x',y')$ is connected to $(y',y')$. Thus $(x,y)$ and $(x',y')$ are in the same connected component of $\Gamma$. Since both were arbitrary, $\Gamma$ is connected and, by Theorem \ref{two agent} we get the desired result. 

Now we are ready to finish the proof. Let $i$ be the agent described in Fact $2$ and $j$ the agent described in Fact $3$. Let $\succsim$ be an arbitrary profile. Let $\succsim_{j}'$ top-rank $\emptyset$. Then $f(\succsim_j',\succsim_{-j})(i)$ is $i$'s top-ranked choice by Fact $2$. Now since the $i,j$-marginal mechanism is a dictatorship by Fact $3$, $j$'s preference cannot affect $i$'s allocation so that $f(\succsim)(i)$ is also $i$'s top-ranked alternative. 

\end{proof}

While to our knowledge this exact result is novel, similar results for have been observed for nearby settings. \citeasnoun{Papai01} proved the result for the special case where $k = H$, that is, when there is no cap on the number of objects an agent can own. \citeasnoun{Hatfield09} proved the result for the related case where each agent can have exactly $k$ houses and no fewer, while our model allows agents to have fewer than $k$ objects.\footnote{To be more precise, both papers show sequential dictatorship is the unique individually strategy-proof, nonbossy, and Pareto efficient mechanism. However, since individual strategy-proofness and nonbossiness are equivalent to group strategy-proofness in our setting, these are the same conditions.}

While the claim itself is closely related to known results, the argument we invoke is different. We reach this conclusion by analyzing the structure of specific two-agents mechanisms. Our main point here is to show how sequential dictatorship is linked across three seemingly disparate settings: social choice, one-sided matching, and multiple assignment. What ties the proofs together across settings is the importance of two-agent marginal mechanisms in understanding the grand mechanism and then invoking our two-agent result to understand the structure of the marginal mechanisms. In all the proofs, a key step is to invoke the structure of the constraint to show that the existence of a dictator for particular two-agent marginal mechanisms implies sequential dictatorship for the grand mechanism.

\section{Conclusion and Future Directions}

Our results provide a technique for studying constrained allocation problems under strong efficiency and incentive properties. While we are able to characterize the full class of mechanisms for many constraints, there are others where an explicit characterization is outstanding. These include school choice and two-sided matching, arguably the most important settings for market design. While we have not explored it here, scheduling problems and many-to-many matching are also promising directions for future research. 

We have focused on group strategy-proofness because of its appealing properties. It is a strong incentive criterion and is equivalent to other common desiderata including Maskin monotonicity (see Appendix section \ref{sec: preliminary obs}). However, we have also seen that in many important settings, no satisfactory group strategy-proof mechanisms exist. For these settings, it is natural to weaken group strategy-proofness to strategy-proofness. We hope to explore this in future work. Analytically, non-bosiness greatly simplifies the analysis. When allowing for bossy mechanisms, as the number of agents gets large, there is room for an agent to manipulate the point on the efficient frontier without affecting her own allocation.

Motivated by the mechanisms found here for two agents, in a companion paper we explore a class of mechanisms we call ``local priority mechanisms" \cite{root2023local}. These are greedy mechanisms which iteratively attempt to give all agents their top choice. Whenever this is not possible, a set of agents is specified as ``local compromisers." This construction enables the development of explicit mechanisms for many constraints. In contrast to the mechanisms described here, one can find constraints for which there are group strategy-proof mechanisms which give no agent their top choice. Local priority mechanisms also generalize many important mechanisms in the literature including trading cycles mechanisms \cite{PyUn17} and deferred acceptance \cite{GaSh62}.

\clearpage

\appendix

\section{Appendix}

\small

\subsection{Preliminary Observations}\label{sec: preliminary obs}

We start with a simple result which recasts strategy-proofness in terms of option sets. 
\begin{lemma}[\citeasnoun{Barbera83}] \label{Barberas Lemma}
A mechanism $f: \pp \rightarrow C$ is strategy-proof if and only if there exist nonempty correspondences $g_i : P^{N-1} \rightrightarrows \obj$ such that, for all agents $i$, $f_i (\succsim) = \max_{\succsim_i} g_i (\succsim_{-i})$
\end{lemma}

\begin{proof}
Define $g_{i}(\succsim_{-i})=f_{i}(P,\succsim_{-i})$ then by strategy-proofness, we have $f_{i}(\succsim_{i},\succsim_{-i})=\argmax_{\succsim_{i}}{g_{i}(\succsim_{-i})}.$ Conversely, if $f$ is defined such that there are $\{g_{i}\}$ as in the proposition, then
\begin{equation*}
f(\succsim_{i},\succsim_{-i})=\argmax_{\succsim_{i}}g_{i}(\succsim_{-i}) \succsim_{i} \argmax_{\succsim_{i}'}g_{i}(\succsim_{-i})=f(\succsim_{i}',\succsim_{-i}).
\end{equation*}
\end{proof}

It will be useful to relate group strategy-proofness with two other notions: nonbossiness and Maskin monotonicity.

\begin{definition}
A mechanism $f:\pp \rightarrow C $ is
\begin{enumerate}
    \item \textbf{nonbossy} if, for all $\succsim\in \pp$,  $$f_{i}(\succsim'_{i},\succsim_{-i})= f_{i}(\succsim) \implies f(\succsim'_{i},\succsim_{-i})= f(\succsim). $$
    \item \textbf{Maskin monotonic} if, for all  $\succsim, \succsim'\in \pp$, 
    $$LC_{\succsim_{i}'}\left[f_{i}(\succsim)\right]\supset LC_{\succsim_{i}}\left[f_{i}(\succsim)\right]\text{  for all }i \implies f(\succsim')=f(\succsim).$$
\end{enumerate}
\end{definition}

\begin{proposition} \label{GSP equivalences}
If $f:\pp\rightarrow \alloc$\ the following are equivalent:
\begin{enumerate}[leftmargin=*]
    \item $f$ is group strategy-proof.
    \item $f$ is strategy-proof and nonbossy.
    \item $f$ is Maskin monotonic.\footnote{Note that the image of $f$ may be arbitrary, so the claim is true for any constraint $C \subset \alloc$.}
\end{enumerate}
\end{proposition}

We can now demonstrate the desired implications for the equivalence in turn:

\begin{proof}

$(1)\implies (2)$: Of course any group strategy-proof mechanism is individually strategy-proof. We will show nonbossiness by contradiction. Suppose there is a profile $\succsim$ and an agent $i$ with an alternative announcement $\succsim_{i}'$ such that $f_{i}(\succsim)=f_{i}(\succsim_{i}',\succsim_{-i})$ but for some $j$, $f_{j}(\succsim)\neq f_{j}(\succsim_{i}',\succsim_{-i})$. Then if $f_{j}(\succsim)\succ_{j} f_{j}(\succsim_{i}',\succsim_{-i})$, the coalition $\{i,j\}$ can improve their outcome at $(\succsim_{i}',\succsim_{-i})$ by announcing $(\succsim_{i},\succsim_{j})$. Otherwise, if  $f_{j}(\succsim)\prec_{j} f_{j}(\succsim_{i}',\succsim_{-i})$, the coalition $\{i,j\}$ can improve their outcome at $\succsim$ by announcing $(\succsim'_{i},\succsim_{j})$.

 $(2)\implies (3)$: Suppose we have two profiles $\succsim, \succsim'\in \pp$  such that
    $$LC_{\succsim_{i}'}\left[f_{i}(\succsim)\right]\supset LC_{\succsim_{i}}\left[f_{i}(\succsim)\right]\text{  for all }i$$
then notice that $f_{1}(\succsim_{1}',\succsim_{2},\dots, \succsim_{n})=f_{1}(\succsim)$ by Lemma \ref{Barberas Lemma} and by nonbossiness we have $f(\succsim_{1}',\succsim_{2},\dots, \succsim_{n})=f(\succsim)$. We can proceed, changing one preference at a time, to show that $f(\succsim')=f(\succsim)$ as desired.

$(3)\implies (1)$: Let $\succsim\in \pp$ and $\succsim'_{A}$ be a candidate deviation for agents in $A$ so that $$f(\succsim'_{A},\succsim_{-A})\succsim_{j}f(\succsim) \text{  for all }j\in A$$
 we will show that this implies $f(\succsim'_{A},\succsim_{-A})=f(\succsim)$.
 For each $j\in A$ construct $\succsim^{*}_{j}$ to be identical to $\succsim_{j}$ except that it ranks $f_{j}(\succsim'_{A},\succsim_{-A})$ first. For any $j\in A$ we have
\begin{equation*}
    \begin{split}
    LC_{\succsim^{*}_{j}}(f_{j}(\succsim'_{A},\succsim_{-A})) &\supset LC_{\succsim_{j}}(f_{j}(\succsim'_{A},\succsim_{-A})) \text{  and} \\
    LC_{\succsim^{*}_{j}}(f_{j}(\succsim)) &\supset LC_{\succsim_{j}}(f_{j}(\succsim)) \\
    \end{split}
\end{equation*}
for all $j$. The first is immediate. To see the second, notice that if $f_{j}(\succsim'_{A},\succsim_{-A})=f_{j}(\succsim)$ then it holds trivially. If instead, $f_{j}(\succsim'_{A},\succsim_{-A})\neq f_{j}(\succsim)$, by assumption we have $f_{j}(\succsim'_{A},\succsim_{-A}) \succ_{j} f_{j}(\succsim)$ and since $\succsim^{*}$ only moves up the position of $f_{j}(\succsim'_{A},\succsim_{-A})$, the second statement holds. However, by Maskin monotonicity, the first statement gives $f(\succsim^{*}_{A},\succsim_{-A})= f(\succsim'_{A},\succsim_{-A})$ and the second gives $f(\succsim^{*}_{A},\succsim_{-A})=f(\succsim)$, so putting them together we get $$f(\succsim'_{A},\succsim_{-A})=f(\succsim^{*}_{A},\succsim_{-A})=f(\succsim)$$ as desired.

\end{proof}

The relationship between group strategy-proofness and Maskin monotonicity was first revealed by the proof of the Muller--Satterthwaite Theorem, which proceeds by showing that either group or individual strategy-proofness is equivalent to Maskin monotonicity for the social choice problem \cite{MuSa77}.\footnote{Recall the Muller--Satterthwaite Theorem: all Maskin monotonic and surjective social choice functions are dictatorial.} This equivalence between group strategy-proofness and Maskin monotonicity was then extended to other problems as well, including to house allocation by \citeasnoun{Svensson99} and for two-sided matching by \citeasnoun{Takamiya01}. \citeasnoun{Takamiya03} unified these observations in a general statement for all indivisible-good economies without externalities that also applies to our model, and should be credited for the equivalence between (1) and (3) in Proposition \ref{GSP equivalences}.

The importance of nonbossiness in relating group and individual incentives was observed for the house allocation problem by \citeasnoun{Papai00}, who proved that individual and group strategy-proofness are equivalent for nonbossy rules. Proposition \ref{GSP equivalences} pushes her observation to more general allocation problems with arbitrary constraints. \citeasnoun{Thomson16} surveys the literature on nonbossiness and its applications, and reviews many specific environments where group and individual incentives coincide. Proposition \ref{GSP equivalences} at this level of abstract generality was also independently proved in the main result of \citeasnoun{Alva17}, who makes a more general observation that the three conditions in Proposition \ref{GSP equivalences} are equivalent for a broad class of preference domains, so we do not claim precedence for the proposition. We mainly present the result here to highlight its importance of this social-choice logic towards establishing our main characterization results. Our proof is also different from that in \citeasnoun{Alva17} because we are not interested in general preference domains, so our argument is consequently more direct and more limited.

\subsection{Proof of Remark \ref{PE image}}

By way of contradiction, suppose that $f:\pp\rightarrow im(f)$ is group strategy-proof and that there is a profile $\succsim$ and an allocation $(a_{i})_{i\in N}\in im(f)$ such that $a_{i}\succsim_{i}f_{i}(\succsim)$ for all $i$ with at least one strict. By definition, there is an alternative profile $\succsim'$ such that $f(\succsim')=(a_{i})_{i\in N}$ which is a profitable deviation from $\succsim$. \qed

\subsection{Proof of Lemma \ref{dumb objects dont matter}}

Let $\{g_{i}\}_{i\in N}$ be as defined in Lemma \ref{Barberas Lemma} of the text. For each $j$ the preference $\succsim_{j}'$ does not change the relative ranking of the objects in $g_{j}(\succsim_{-j})$ hence we have $f_{j}(\succsim_{j}',\succsim_{-j})=f_{j}(\succsim)$ by Lemma \ref{Barberas Lemma} so by nonbossiness $f(\succsim_{j}',\succsim_{-j})=f(\succsim)$. Repeating this argument one agent at a time gives the result.  \qed

\subsection{Proof of Theorem \ref{single-compromising}}

We show that every group strategy-proof and Pareto efficient mechanism is of the form $f^\alpha$ for some adapted local compromiser assignment $\alpha$. Let $C$ be a unilateral constraint and fix and a group strategy-proof, efficient mechanism $g:\pp \rightarrow C$. Let $a=(a_{i})_{i\in N}$ be infeasible. For every $i$ there is an object $a_{i}'$ such that $(a_{i}',a_{-i})\in C$. Let $\succsim_{i}\in P^{\uparrow}\left[a_{i},a_{i}'\right]$ for each $i$. Since $g$ always returns feasible allocations, there is at least one agent $k$ who doesn't get their top choice at the constructed preference profile $\succsim=(\succsim_{i})_{i\in N}$. However, Pareto-efficiency then implies that $g_{i}(\succsim)=a_{i}$ for all $i\neq k$ and $g_{k}(\succsim)=a_{k}'$. By Maskin monotonicity and Lemma \ref{Barberas Lemma} we have that for any $\succsim_{-k}'$ with $\max_{\succsim'_{j}}\obj=a_{j}$ for all $j\neq k$, $a_{k}\notin g_{k}(\succsim_{-k})$, so that $k$ always compromises when the top choice is $a$. Define $\alpha(a)=k$ (we can do this unambiguously because no other agent always compromises at $a$, e.g. at the profile $\succsim$). Since $a$ was an arbitrary infeasible allocation, we can do the same for any other infeasible allocation to define $\alpha$ on all of $\bar{C}$. Finally, we establish inductively that $f$ is local priority according to $\alpha$. Pick any preference profile $\succsim'$. Start at $a^{1}=(\max_{\succsim_{i}'}\obj)_{i \in N}$. If this is feasible, then $f$ being Pareto efficient implies $g(\succsim')=a^{1}$. Otherwise, it is infeasible, and by the previous argument, we have an agent $k=\alpha(a^{1})$ who must compromise. Replace $\succsim'_{k}$ with the same preference, except that it puts $a^{1}_{k}$ last. By Maskin monotonicity, this cannot affect the outcome of $f$. We therefore repeat the above process at the new profile. This is exactly how the local priority mechanism according to $\alpha$ works, giving the result. 

Now we need to show that $\alpha$ has to satisfy the property that if $\alpha(a)=i$ then for any $(a_{i}',a_{-i})\in \bar{C}$, we have $\alpha(a_{i}',a_{-i})=\{i\}$. However this follows from similar reasoning as in the two-agent case. If, instead $k=\alpha(a_{i}',a_{-i})$ consider the profile $\succsim$ with $\tau(\succsim)=a$ and $\tau_{2}(\succsim_{i})=a_{i}'$ and $\tau_{2}(\succsim_{k})=a_{k}'$ where $(a_{k}',a_{-k})\in C$. We get a violation of Pareto efficiency since the local priority algorithm would make both $i$ and $k$ compromise to their second-best choice, which would be Pareto dominated by $(a_{k}',a_{-k})$.

The fact that this mechanism is group strategy-proof and Pareto efficient is now a simple consequence of Maskin monotonicity and Remark \ref{PE image}. \qed

\subsection{Proof of Lemma \ref{N agent} (N-agent characterization)}

If $f$ is group strategy-proof, the marginal mechanisms are group strategy-proof by definition. For the other direction, suppose that every two-agent marginal mechanism is group strategy-proof. By Proposition \ref{GSP equivalences} it suffices to show that $f$ is individually strategy-proof and nonbossy. Then $f$ is individually strategy-proof since for any $i$ and any profile $\succsim$ we can choose $j\neq i$ and consider the marginal mechanism $f^{i,j}_{\succsim_{-i,j}}$ then in this marginal mechanism $i$ cannot profit from misreporting, hence she cannot in $f$. Now suppose we have $f_{i}(\succsim_{i}',\succsim_{-i})=f_{i}(\succsim)$ and for some $j$, $f_{j}(\succsim_{i}',\succsim_{-i})\neq f_{j}(\succsim)$, either $f_{j}(\succsim_{i}',\succsim_{-i})\succ_{j} f_{j}(\succsim)$ or $f_{j}(\succsim_{i}',\succsim_{-i})\prec_{i} f_{j}(\succsim)$. However, by assumption the marginal mechanism $f^{ij}_{\succsim_{-ij}}$ is group strategy-proof. From the two-agent characterization, no two-agent group strategy-proof mechanism can have this property. \qed

\subsection{Proof of Remark \ref{generalized_serial_dictatorships}}

By Proposition \ref{GSP equivalences}, it suffices to show that sequential dictatorships are strategy-proof and nonbossy. It is clear that for any $\zeta$ the sequential dictatorship for $\zeta$ is individually strategy-proof. Since $\zeta$ only depends on the allocations of earlier dictators, it is also nonbossy. 

To see that it is Pareto efficient, by Remark \ref{PE image} it is enough to establish that its image is exactly $C$. By construction, the image is a subset of $C$. For any feasible allocation $a\in C$ let $\succsim_{i}$ put $a_{i}$ first for all $i$. Then $f(\succsim)=a$ so $im(f)=C$. \qed

\subsection{Proof of Corollary \ref{cor: existence}}
The projection of the house allocation constraint is shown in Figure \ref{two-agent examples}. If there are at least two houses, this can easily be seen to have a graph $\Gamma$ with at least two connected components. 

For the two-sided matching problem, pick a pair of agents $m_1\in M$ and $w_1\in W$. Then $\Gamma(C^{\{w_1,m_1\}})$ has two connected components, namely the infeasible allocations $(m_1,w_k)$ where $k\neq 1$ and the infeasible allocations $(m_l,w_1)$ with $l\neq 1$. 

In the school choice problem, with schools $s$ and $t$ such that $q_s+q_t\leq N$. Construct a non-dictatorial mechanism $f$ as follows. First run serial dictatorship where agents $1,2,\dots, N-2$ pick in order of their index. If the suballocation does not result in exactly one seat remaining at both $s$ and $t$, have agents $N-1$ and $N-2$ pick in order. If $s$ and $t$ both have exactly one seat remaining, use one of the non-dictatorial strategy-proof and efficient mechanisms to match $N-1$ and $N-2$. This construction will result in a non-bossy and strategy-proof mechanism which is group strategy-proof and efficient by Proposition \ref{GSP equivalences} and Remark \ref{PE image}.

\subsection{Proof of Lemma \ref{GSP=ISP for social choice}}

By Proposition \ref{GSP equivalences}, it is enough to show that any strategy-proof mechanism is non-bossy. However, nonbossiness is immediate since all agents are allocated the same object.
\qed

\subsection{Proof of Theorem \ref{GS Theorem} (Gibbard--Satterthwaite Theorem)}

Let $C$ be the diagonal (i.e. the social choice constraint) and $\vert\obj\vert\geq 3$.

From Proposition \ref{GSP equivalences}, it suffices to show that any group strategy-proof mechanism is dictatorial. Note that non-bosiness is immediate for this constraint. We will establish the result in two steps. First, we will show that for some $i,j$ and some profile $\succsim_{-ij}=(\succsim_{k})_{k\neq i,j}$ we have $\vert I^{ij}(\succsim_{-ij})\vert \geq 3$. From the characterization of two-agent mechanisms, we will see that $f^{ij}_{\succsim_{-ij}}$ is dictatorial. We will then show that this implies the entire mechanism is dictatorial.  
\begin{enumerate}
    \item Suppose by way of contradiction that for all $i,j$ and all $\succsim_{-ij}$ we have $\vert I^{ij}(\succsim_{-ij})\vert <3 $. First, note that if for all $i,j$ and all $\succsim_{-ij}$ we have $\vert I^{ij}_{\succsim_{-ij}}\vert =1$ then $f$ is single-valued which contradicts the surjectivity of $f$.\footnote{To see that $f(\succsim)=f(\succsim')$, change one preference at a time. No single change can alter $f$, so we get the result.} Hence there is at least one pair of agents $i,j$ and $\succsim_{-ij}$ such that $\vert I^{ij}(\succsim_{-ij})\vert \geq 2$. For simplicity and without loss, let $i=1$ and $j=2$. By assumption then $\vert I^{ij}(\succsim_{-ij})\vert = 2$ and without loss assume $I^{ij}(\succsim_{-ij})=\{a,b\}$. Then there must be a local dictator assigned to the incompatible pairs $(a,b)$ and $(b,a)$. This leaves (up to symmetry) two marginal mechanisms $\phi_{1}$ and $\phi_{2}$ where $$\phi_{1}(\succsim_{1},\succsim_{2})=
    \begin{cases}
    a &\text{ if } a\succ_{1} b \\
    b &\text{ if } a\prec_{1} b 
    \end{cases}$$
    and 
    $$\phi_{2}(\succsim_{1},\succsim_{2})=
    \begin{cases}
    a &\text{ if } a\succ_{1} b \text{ and } a\succ_{2} b \\
    b &\text{ otherwise } 
    \end{cases}$$
    In the first, agent $1$ is a dictator. In the second, $b$ is chosen by default and $a$ is only chosen if both agents prefer it to $b$. Let $c$ be another object in $\mathscr{O}$. If we let $\succsim_{2}^{*}\in \mathscr{P}^{\uparrow}[c,a,b]$ then in either case we have $f(\succsim_{1},\succsim^{*}_{2},\succsim_{-1,2})=a$ if $a\succ_{1}b$ and $f(\succsim_{1},\succsim^{*}_{2},\succsim_{-1,2})=b$ if $b\succ_{1}a$. We then have that $a$ and $b$ are in $I^{1,3}(\succsim_{2}^{*},\succsim_{4},\dots,\succsim_{n})$. As before we have two possible mechanisms and in either one, if $\succsim_{3}^{*}\in \mathscr{P}^{\uparrow}[c,a,b]$ we have $f(\succsim_{1},\succsim^{*}_{2},\succsim_{3}^{*},\succsim_{4},\dots,\succsim_{n})=a$ if $a\succ_{1}b$ and $f(\succsim_{1},\succsim^{*}_{2},\succsim_{3}^{*},\succsim_{4},\dots,\succsim_{n})=b$ if $b\succ_{1}a$. Continuing in this way, we get a profile of preferences in which all agents prefer $c$, but $c$ is not chosen. Since any group strategy-proof map is efficient on its image we must either have that $c\notin im(f)$ or $f$ is not group strategy-proof. Either way we have a contradiction.
    
    \item From the characterization of two-agent mechanisms, if $|I^{1,2}(\succsim_{-1,2})|\geq 3$ we have a single dictator in the marginal mechanism $f^{ij}_{\succsim_{-ij}}$. For simplicity let $i=1,j=2$ and assume $1$ is the dictator. We will show that for any $\succsim'$, $f(\succsim')=\max_{\succsim'_{1}}I^{1,2}(\succsim'_{-1,2})$. Begin with $f(\succsim'_{1},\succsim_{2},\dots,\succsim_{n})$. The statement holds by assumption. Now since $1$ is the marginal dictator, changing $\succsim_{2}$ to $\succsim_{2}'$ cannot change the outcome. Hence the statement holds for $f(\succsim'_{1},\succsim'_{2},\dots,\succsim_{n})$. Now we have that $I^{1,3}(\succsim_{2}',\succsim_{4},\dots,\succsim_{n})$ contains $I^{1,2}(\succsim_{-1,2})$ as a subset. Hence there either $1$ or $3$ is a local dictator. Clearly it must be $1$. Therefore $3$'s announcement cannot change the outcome, so we have $f(\succsim'_{1},\succsim'_{2},\succsim_{3}',\succsim_{4},\dots,\succsim_{n})=\max_{\succsim'_{1}}I^{1,2}(\succsim_{-1,2})$. Continuing in this way gives the desired result. The assumption that $f$ is surjective implies that $1$ is a dictator. \qed
\end{enumerate}

\subsection{Proof of Theorem \ref{Roommates_Characterization} (Roommates characterization)}
The ``if" direction follows directly from Remark \ref{generalized_serial_dictatorships}.

We will prove the ``only if" direction by mathematical induction. First, by Lemma \ref{dumb objects dont matter}, we may ignore any agents' ranking for infeasibly matching with herself. If $N=2$ there is only one feasible allocation, so every mechanism is trivially a sequential dictatorship. If $N=4$, then the problem is a social choice problem since a single agent's match determines the matches of everyone else. In this case, the result follows from the Gibbard--Satterthwaite Theorem (Theorem \ref{GS Theorem}). Suppose that for all $m<n$ when there are $2m$ agents, all group strategy-proof and Pareto efficient roommates mechanisms are sequential dictatorships. We will show this for $2n$ agents. It will be enough to show that there is an agent $j$ such that $f_{j}(\succsim)=\argmax_{\succsim_{j}}N$ for all $\succsim$, since, conditional on each of $j$'s choices, the remaining $2n-2$ agents need to assigned a roommate and since $f$ is group strategy-proof and Pareto efficient, this marginal mechanism will be group strategy-proof and Pareto efficient so the induction assumption will guarantee that it is a sequential dictatorship.

Let $f$ be a group strategy-proof and Pareto efficient roommates mechanism for $2n$ agents with $n\geq 3$. We will first consider the possible two-agent marginal mechanisms. Let $i\neq j$ and fix a profile $\succsim_{-ij}$ of the other agents. Assume $(j,i)\in I^{ij}(\succsim_{-ij})$, so that it is possible for $i$ and $j$ to match when the other agents announce $\succsim_{-ij}$. For all $k\neq i$, $(j,k)\notin I^{ij}(\succsim_{-ij})$ since $(j,k)$ has $i$ matched to $j$ but $j$ matched to $k$. Likewise, for all $k\neq j$ we have $(k,i)\notin I^{ij}(\succsim_{-ij})$. Define $\chi_{i}=\{x\in N \sst (x,y)\notin I^{ij}(\succsim_{-ij})  \text{ for all }y\in N\}$ and $\chi_{j}=\{y\in N \sst (x,y)\notin I^{ij}(\succsim_{-ij}) \text{ for all }x\in N\}$. Then after possibly permuting the rows and columns, we get a marginal constraint as illustrated in the two panels of figure \ref{roommates_1} in the main text. 
\noindent As usual, we will ignore agents preferences over objects they can never receive.\footnote{In this case, $\chi_{i}$ and $\chi_{j}$ are not possible for $i$ and $j$ to match holding fixed the preferences $\succsim_{-j}$.} If $\left[N-\chi_{i}\cup\{j\}\right]\times \left[N-\chi_{j}\cup\{i\}\right]$ intersects any infeasible point, then the graph $G(I^{ij}(\succsim_{-ij}))$ is totally connected, as illustrated on the right-hand picture of figure \ref{roommates_1}.\footnote{Recall the relation graph $G(C)$ was defined for every two-agent constraint in section \ref{two agent section}.} Therefore there must be a single dictator in the marginal mechanism $f^{ij}_{\succsim_{-ij}}$ by Theorem \ref{two agent} and Lemma \ref{N agent}. Otherwise, every allocation in $\left[N-\chi_{i}\cup\{j\}\right]\times \left[N-\chi_{j}\cup\{i\}\right]$ is feasible or the set is empty. In the latter case $I^{ij}(\succsim_{-ij})$ is a singleton, and obviously there is only one marginal mechanism. In the former case, as a consequence of Theorem \ref{two agent} there are four possible Pareto efficient, strategy-proof marginal mechanisms as illustrated in figure \ref{roommates_2} in the main text.


In panel (A), $j$ is the dictator since $i$ must compromise at every infeasible allocation. In panel (B), $i$ is the dictator. In Panel (C), $i$ and $j$ are matched together if either top-ranks the other and are only unmatched if both $i$ prefers someone in $N-\chi_{i}\cup\{j\}$ and $j$ prefers someone in $N-\chi_{j}\cup\{i\}$. In panel (D), $i$ and $j$ are matched only if both top-rank the other and are unmatched otherwise.

Summarizing, if $(j,i)\in I^{ij}(\succsim_{-ij})$, there are four possible types of mechanisms $f^{ij}_{\succsim_{-ij}}$:

\begin{enumerate}
    \item $f^{ij}_{\succsim_{-ij}}$ is constant and $(j,i)$. In this case, $N-\chi_{i}=\{j\}$ and $N-\chi_{j}=\{i\}$.
    \item $f^{ij}_{\succsim_{-ij}}$ is dictatorial, so $i$ gets their top choice from $N-\chi_{i}$ or $j$ gets their top choice from or $N-\chi_{j}$ and the other agent gets their top choice consistent with the dictators' allocation. Note that in a dictatorial mechanism, the non-dictator cannot affect the option set of the dictator. 
    \item $i$ and $j$ are matched by default, and are unmatched only if both agree. This is shown in panel (C). In this case, all allocations in $\left[N-\chi_{i}\cup\{j\}\right]\times \left[N-\chi_{j}\cup\{i\}\right]$ are feasible.
    \item $i$ and $j$ are unmatched by default and are matched only if both agree. This is shown in Panel (D). In this case, all allocations in $\left[N-\chi_{i}\cup\{j\}\right]\times \left[N-\chi_{j}\cup\{i\}\right]$ are feasible.
\end{enumerate}

In the remainder of the proof, we will often need to show that a given two-agent marginal mechanism is dictatorial. To do that, it will be sufficient to show that it is possible for both agents to match with one another, that it is non-constant (i.e. that there are at least two possible allocations for the two agents holding the other agents' preferences fixed), and that it is not of the third or fourth types. The third type of mechanism is usually easy to rule out. If we can find a preference where one agent top-ranks the other and they are still not matched, it cannot be of type three. Type (4) is somewhat more subtle, but we can rule it out if an agent can match with a second agent even when that agent bottom-ranks the first agent.

Note that if we partition the set of agents into two nonempty even sets $A$ and $N\setminus A$ and if we restrict attention to preferences where the agents in $A$ rank all agents in $A$ over all agents in $N\setminus A$ and likewise the agents in $N\setminus A$ rank themselves above the agents in $A$, then by Pareto efficiency for all such preferences, agents in $A$ are matched to themselves and agents in $N\setminus A$ are matched within their own group. The induction assumption implies that both groups are matched using a sequential dictatorship. The next lemma (whose validity depends on the induction assumption) says that the dictator on either side retain their dictatorship rights if the other agents on their side switch to an arbitrary preference.

\begin{lemma}\label{roommates_lemma}
Let $A$ be a nonempty proper subset of $N$ with an even number of agents and $\vert A \vert\geq 4$. If $\succsim_{N\setminus A}\in \left[P^{\uparrow}(N\setminus A)\right]^{N\setminus A}$, then there is an agent $j\in A$ such that $$f_{j}(\succsim_{A},\succsim_{N\setminus A})=\argmax_{\succsim_{j}}N$$ whenever $\argmax_{\succsim_{j}}N\in A$. Equivalently, $g_{j}(\succsim_{A-\{j\}},\succsim_{N\setminus A})\supset A-\{j\}$ for all $\succsim_{A-\{j\}}$.
\end{lemma}

\begin{proof}
For notational convenience, let $A=\{1,2,\dots,l\}$ and $N\setminus A = \{l+1,\dots N\}$. Fix a profile $\succsim_{N\setminus A}\in \left[P^{\uparrow}(N\setminus A)\right]^{N \setminus A}$. For any $\succsim_{1}, \dots \succsim_{l} \in P^{\uparrow}(\{1,2,\dots, l\})$, by Pareto efficiency, $f(\succsim_{1},\succsim_{2},\dots, \succsim_{l},\succsim_{N\setminus A})$ will match agents in $\{1,2,\dots, l\}$ with other agents in $\{1,2,\dots, l\}$ and agents in $\{l+1,\dots N\}$ with other agents in $\{l+1,\dots N\}$. Thus the marginal mechanism $f(\cdot,\succsim_{N\setminus A})$ restricted to profiles in $\left[P^{\uparrow}(\{1,2,\dots, l\})\right]^{l}$ gives a roommates mechanism for the agents in $\{1,2,\dots, l\}$. By the group strategy-proofness and efficiency of $f$, the marginal mechanism is also group strategy-proof and efficient. By the induction assumption this marginal mechanism is a sequential dictatorship. Without loss, assume that $1$ is the first dictator. Then we have $g_{1}(\succsim_{2},\dots, \succsim_{l},\succsim_{N\setminus A}) \supset \{2,3,\dots l\}$ for all $\succsim_{2},\dots, \succsim_{l}$ in $P^{\uparrow}(\{1,2,\dots, l\})$. For any $\succsim_{3},\dots, \succsim_{l}$ in $P^{\uparrow}(\{1,2,\dots, l\})$, consider the $1,2$-marginal mechanism. Since $g_{1}(\succsim_{2},\dots, \succsim_{l},\succsim_{N\setminus A})\supset \{2,3,\dots, l\}$ for all $\succsim_{2},\dots, \succsim_{l}$ in $P^{\uparrow}(\{1,2,\dots, l\})$, if $1$ top-ranks $2$ and $2$ announces any preference in $P^{\uparrow}(\{1,2,\dots, l\})$, $1$ and $2$ are matched. Thus $(2,1)\in I^{1,2}(\succsim_{3},\dots, \succsim_{l},\succsim_{N\setminus A}).$ From the considerations above, there are four possibilities for this marginal mechanism. Let $\succsim^{*}_{1}$ top rank $j\neq 2$ and $j\leq l$ and $\succsim_{2}$ in $P^{\uparrow}(\{1,2,\dots, l\})$ top-rank $1$. At this profile, $1$ and $j$ are matched. Hence the $1,2$ marginal mechanism is not constant. Furthermore, it cannot be of type (3), since $1$ is matched with $j$, despite $2$ top-ranking $1$. Let $\succsim_{2}^{*}$ be in $P^{\uparrow}(\{1,2,\dots, l\})$ and top-rank her match at the profile $(\succsim^{*}_{1},\succsim_{2})$. Since $1$ and $2$ are matched when $1$ top-ranks $2$ and $2$ announces $\succsim_{2}^{*}$, the mechanism also cannot be of type (4) (At $\succsim_{2}^{*}$, agent $2$ is top-ranking a feasible match in the $1,2$ marginal mechanism, but $1$ can still match with her). The only possibility left is that the $1,2$-marginal mechanism is dictatorial with agent $1$ as the dictator. Since non-dictators cannot affect the option set of dictators, we get that $g_{1}(\succsim_{2}', \succsim_{3}, \dots, \succsim_{l},\succsim_{N\setminus A})\supset \{2,3,\dots, l\}$ for any $\succsim_{2}'$ and any $\succsim_{3}, \dots, \succsim_{l}$ in $P^{\uparrow}(\{1,2,\dots, l\})$. We could have carried out the above argument with any $i$ in place of $2$, so in fact we have 
\begin{equation*}
g_{1}(\succsim_{2}, \dots, \succsim_{i-1},\succsim_{i}',\succsim_{i+1}, \dots \succsim_{l},\succsim_{N\setminus A})\supset \{2,3,\dots l\}
\end{equation*}
for any $\succsim_{i}'$ and any $\succsim_{2}, \dots, \succsim_{i-1},\succsim_{i+1}, \dots \succsim_{l}$ in $P^{\uparrow}(\{1,2,\dots, l\})$.

The goal is to show that 
\begin{equation*}
g_{1}(\succsim_{2}', \dots, \succsim_{l}',\succsim_{N\setminus A})\supset \{2,3,\dots l\}
\end{equation*}
for all $\succsim_{2}', \dots, \succsim_{l}'$. We will do this by induction. Specifically we will show that if for any $0<q-1< l-1$ and any $A'\subset A-\{1\}$ with $\vert A'\vert=q-1$ we have $g_{1}(\succsim_{A'}',\succsim_{A-A'\cup \{1\}},\succsim_{N\setminus A})\supset \{2,3,\dots l\}$ for any $\succsim_{A'}'$ and any $\succsim_{A-A'\cup \{1\}}$ in $\left[P^{\uparrow}(A)\right]^{A-A'\cup \{1\}}$ then the same holds for any $A'\subset A-\{1\}$ with $q$ agents.

For simplicity, let $A'=\{2,\dots q+1\}$ and pick any $\succsim_{2}',\dots,\succsim_{q+1}'$. By the induction assumption, we have $g_{1}(\succsim_{2}',\dots,\succsim_{q}',\succsim_{q+1},\dots \succsim_{l},\succsim_{N\setminus A})\supset \{2,3,\dots l\}$ for any $\succsim_{2}',\dots,\succsim_{q}'$ and any $\succsim_{q+1},\dots \succsim_{l}$ in $P^{\uparrow}(A)$. Now by the same arguments as above, the $1,q+1$-marginal mechanism at this profile is either of type (2) (i.e. dictatorial) or it is of type (4). Suppose, by way of contradiction, that it is of type (4) and let $\succsim_{q+1}^{*}$ bottom-rank $1$. Then doing so removes $q+1$ from $1$'s option set, but leaves it otherwise the same. Let $\succsim_{1}^{**}$ top-rank $q+1$ and second-rank $q$. From the above discussion, we get that $1$ is matched to $q$ at the marginal profile $(\succsim_{1}^{**},\succsim_{q+1}^{*})$. If we let $\succsim_{q}^{*}\in P^{\uparrow}(A)$ top-rank $1$, then by Maskin-monotonicity, we have 
\begin{equation*}
    f(\succsim_{1}^{**},\succsim_{2}',\dots, \succsim_{q}',\succsim_{q+1}^{*}, \succsim_{q+2},\dots , \succsim_{l},\succsim_{N\setminus A})= \\
    f(\succsim_{1}^{**},\succsim_{2}',\dots, \succsim_{q-1}',\succsim_{q}^{*},\succsim_{q+1}^{*}, \succsim_{q+2},\dots , \succsim_{l},\succsim_{N\setminus A})
\end{equation*}
but on the left we have $1$ is matched to $q$, her second-top choice. By the induction assumption, on the right we should have $q+1$ in $1$'s option set since the agents $q, q+2,\dots, l$ are all announcing a preference in $P^{\uparrow}(A)$, leaving only $q-1$ agents announcing a possibly different preference. This gives a contradiction so we must have that $1$ is the dictator in the $1,q+1$-marginal mechanism.
\end{proof}

We will call agent $j$ in the lemma above, the \textit{marginal dictator}. Having done this, the idea is to partition the agents in two ways. First we consider the partition $\{1,2\}\{3,4,\dots,N\}$. By Lemma \ref{roommates_lemma}, for $\succsim_{1}^{*}\in P^{\uparrow}(2)$ and $\succsim_{2}^{*}\in P^{\uparrow}(1)$ there is a marginal dictator among $\{3,4,\dots,N\}$ which we can assume without loss is agent $3$. Second, we consider the partition $\{1,2,3,4\},\{5,6,\dots,N\}$ and again Lemma \ref{roommates_lemma} says that given $\succsim_{5}^{*},\dots, \succsim_{n}^{*} \in P^{\uparrow}(\{5,\dots,n\})$ , there is a marginal dictator among $\{1,2,3,4\}$. We show that by comparing these two dictators, we can find a single dictator for the whole mechanism.

As above, let $\succsim_{1}^{*}\in P^{\uparrow}(2)$, $\succsim_{2}^{*}\in P^{\uparrow}(1)$ and without loss assume that $3$ is the marginal dictator among $\{3,\dots N\}$. By Maskin-monotonicity, it is also without loss to suppose that both $\succsim_{1}^{*}$ and  $\succsim_{2}^{*}$ bottom-rank $3$. \footnote{Let $\succsim_{1}^{*}\in P^{\uparrow}(2)$, $\succsim_{2}^{*}\in P^{\uparrow}(1)$, by Lemma \ref{roommates_lemma}, we have $g_{3}(\succsim_{1}^{*},\succsim_{2}^{*},\succsim_{4}',\dots , \succsim_{N}')\supset \{4,\dots, N\}$ Let $\succsim_{1}^{**}$ and $\succsim_{2}^{**}$ be the same as $\succsim_{1}^{*}$ and $\succsim_{2}^{*}$ respectively, except both bottom-rank $3$. Let $\succsim_{3}$ top rank $k\in \{4,\dots, N\}$. Then $f_{3}(
\succsim_{1}^{*},\succsim_{2}^{*},\succsim_{3},\succsim_{4}',\dots, \succsim_{N}')=k$ for any $\succsim_{4}',\dots, \succsim_{N}'$. But Maskin-monotonicity then says $f_{3}(
\succsim_{1}^{**},\succsim_{2}^{**},\succsim_{3},\succsim_{4}',\dots, \succsim_{N}')=k$ for any $\succsim_{4}',\dots, \succsim_{N}'$.} Also choose $\succsim_{5}^{*},\dots, \succsim_{n}^{*} \in P^{\uparrow}(\{5,\dots,n\})$. By Lemma \ref{roommates_lemma}, $g_{3}(\succsim_{1}^{*},\succsim_{2}^{*},\succsim_{4}',\dots , \succsim_{N}')\supset \{4,\dots, N\}$ for all $\succsim_{4}',\dots , \succsim_{N}'$. Likewise, for some $i\in\{1,2,3,4\}$, we have $g_{i}(\succsim_{\{1,2,3,4\}-\{i\}}',\succsim_{5}^{*},\dots, \succsim_{n}^{*})\supset \{1,2,3,4\}-\{i\}$ for all $\succsim_{\{1,2,3,4\}-\{i\}}'$. This gives four cases, corresponding to the possible identities of $i$. However, note that $i$ cannot be $4$ since $3$ and $4$ are matched at the profile $(\succsim_{1}^{*},\succsim_{2}^{*},\succsim_{3},\succsim_{4},\succsim_{5}^{*},\dots, \succsim_{n}^{*})$ where $3$ top ranks $4$ regardless of $\succsim_{4}$. Since $1$ and $2$ are so far symmetric, this leaves two cases: $i=1$ (and $i=2$ by symmetry) and $i=3$. 

We will start with the latter case. So we have 
\begin{equation} \label{eq1}
    g_{3}(\succsim_{1}^{*},\succsim_{2}^{*},\succsim_{4}',\dots , \succsim_{N}')\supset \{4,\dots, N\} \text{ for all } \succsim_{4}',\dots , \succsim_{N}'\text{, and}
\end{equation}
\begin{equation}\label{eq2}
    g_{3}(\succsim_{1}',\succsim_{2}',\succsim_{4}',\succsim_{5}^{*},\dots, \succsim_{n}^{*})\supset \{1,2,4\}\text{ for all } \succsim_{1}',\succsim_{2}',\succsim_{4}'
\end{equation}
In particular, $g_{3}(\succsim_{1}^{*},\succsim_{2}^{*},\succsim_{4}', \succsim_{5}^{*}\dots , \succsim_{N}^{*})= N-\{3\}$ for all $\succsim_{4}'$. We need to show $g_{3}(\succsim_{1}',\succsim_{2}',\succsim_{4}', \succsim_{5}'\dots , \succsim_{N}')= N-\{3\}$ for all $(\succsim_{1}',\succsim_{2}',\succsim_{4}', \succsim_{5}'\dots , \succsim_{N}')$. Consider the $3,5$-marginal mechanism at the profile $(\succsim_{1}^{*},\succsim_{2}^{*},\succsim_{4}', \succsim_{6}^{*}\dots , \succsim_{N}^{*})$ for any $\succsim_{4}'$. From equation \ref{eq1} above, $3$ and $5$ are matched whenever $3$ top ranks 5, regardless of $5$'s preference. It is also possible for $3$ to match with $4$ regardless of $5$'s preference. From the discussion about the possible two-agent marginal mechanisms, the only possibility for this marginal mechanism has $3$ as the dictator. In this case, $5$'s announcement cannot affect $3$'s option set. Thus we have $g_{3}(\succsim_{1}^{*},\succsim_{2}^{*},\succsim_{4}', \succsim_{5}',\succsim_{6}^{*}\dots , \succsim_{N}^{*})= N-\{3\}$ for any $\succsim_{4}',\succsim_{5}'$. Repeating this argument one agent at a time implies that \begin{equation} \label{eq3}
    g_{3}(\succsim_{1}^{*},\succsim_{2}^{*},\succsim_{4}',\dots , \succsim_{N}') = N-\{3\} \text{ for all } \succsim_{4}',\dots , \succsim_{N}'\text{, and}
\end{equation}
a symmetric argument shows that
\begin{equation} \label{eq4}
    g_{3}(\succsim_{1}',\succsim_{2}',\succsim_{4}',\succsim_{5}^{*},\dots, \succsim_{n}^{*}) = N-\{3\} \text{ for all } \succsim_{1}',\succsim_{2}',\succsim_{4}'.
\end{equation}

Now we will use equation \ref{eq3} to get the desired result. We will do this by looking at the $1,3$ and $2,3$ marginal mechanisms. Equation \ref{eq3} already implies that these mechanisms cannot be type (4) since $3$ can match with $1$ and $2$ even though both bottom-rank her. The main thing to do is show that the marginal mechanisms are not of type (3). To do this, we will need to show that when $1$ or $2$ switch to top-ranking $3$ they do not force a match.

Let $\succsim_{1}^{**}$ be identical to $\succsim_{1}^{*}$, except that $3$ is top ranked. Define $\succsim_{2}^{**}$ equivalently. Now we want to show that the following three equations hold:
\begin{equation} \label{eq5}
    g_{3}(\succsim_{1}^{**},\succsim_{2}^{*},\succsim_{4}',\dots , \succsim_{N}') = N-\{3\} \text{ for all } \succsim_{4}',\dots , \succsim_{N}'\text{, and}
\end{equation}
\begin{equation} \label{eq6}
    g_{3}(\succsim_{1}^{*},\succsim_{2}^{**},\succsim_{4}',\dots , \succsim_{N}') = N-\{3\} \text{ for all } \succsim_{4}',\dots , \succsim_{N}'\text{, and}
\end{equation}
\begin{equation} \label{eq7}
    g_{3}(\succsim_{1}^{**},\succsim_{2}^{**},\succsim_{4}',\dots , \succsim_{N}') = N-\{3\} \text{ for all } \succsim_{4}',\dots , \succsim_{N}'.
\end{equation}
Since the arguments are all symmetric, we will just show equation \ref{eq5}. From equation \ref{eq4}, we know that  $g_{3}(\succsim_{1}^{**},\succsim_{2}^{*},\succsim_{4}',\succsim_{5}^{*}, \dots , \succsim_{N}^{*}) = N-\{3\}$. Consider the $3,5$-marginal mechanism at the profile $(\succsim_{1}^{**},\succsim_{2}^{*},\succsim_{4}',\succsim_{6}^{*}, \dots , \succsim_{N}^{*})$. Let $\succsim_{3}\in P^{\uparrow}(5,4)$ and $\succsim_{3}''\in P^{\uparrow}(4,5)$. Then we have that $3$ and $5$ are matched at the profile $(\succsim_{3},\succsim_{5}^{*})$ and $3$ and $4$ are matched at the profile $(\succsim_{3}'',\succsim_{5}^{*})$. Thus the $3,5$-marginal mechanism is not constant and if it is dictatorial, $3$ must be the dictator. We must also rule out type (3) and type (4) mechanisms. Let $\succsim_{5}''$ top-rank $3$. In a type (3) mechanism, we would have that $3$ and $5$ are matched at the profile $(\succsim_{3}'',\succsim_{5}'')$. However, going back to the full mechanism, this would imply, by Maskin monotonicity that $$f(\succsim_{1}^{**},\succsim_{2}^{*},\succsim_{3}'',\succsim_{4}',\succsim_{5}'',\succsim_{6}^{*},\dots,\succsim_{N}^{*})=f(\succsim_{1}^{*},\succsim_{2}^{*},\succsim_{3}'',\succsim_{4}',\succsim_{5}'',\succsim_{6}^{*},\dots,\succsim_{N}^{*})$$
however, on the right hand side, we have $3$ matched with $4$ by equation \ref{eq3}. Thus the $3,5$-marginal mechanism cannot be of type (3). Finally, suppose that $\succsim_{5}'''$ ranks agent $3$ last. If the marginal mechanism were type (4), we could not have $3$ and $5$ matched at $(\succsim_{3},\succsim_{5}''')$. However, in a type (4) mechanism, either agent can only remove themselves from the other agents option set. Hence in this case we would have that $3$ is matched to $4$ at $(\succsim_{3},\succsim_{5}''')$. But for the same reasons as above, Maskin monotonicity implies this cannot happen. Hence $3$ is the dictator in the marginal mechanism and $5$'s preference does not affect $3$'s option set so 
\begin{equation*}
    g_{3}(\succsim_{1}^{**},\succsim_{2}^{*},\succsim_{4}', \succsim_{5}', \succsim_{6}^{*}, \dots , \succsim_{N}^{*}) = N-\{3\} \text{ for all } \succsim_{4}',\succsim_{5}'.
\end{equation*}
Repeating this argument one agent at a time gives us equation \ref{eq5}.

Now we claim that equations \ref{eq3} and \ref{eq6}, together imply that 
\begin{equation} \label{eq8}
    g_{3}(\succsim_{1}^{*},\succsim_{2}',\succsim_{4}',\dots , \succsim_{N}') = N-\{3\} \text{ for all } \succsim_{2}',\succsim_{4}',\dots , \succsim_{N}'
\end{equation}
Equation \ref{eq3} says that $3$ has the option to match with $2$, even though $2$ bottom-ranks $3$ by assumption. Equation \ref{eq6} that $3$ has the option to not match with $2$, even if $2$ top ranks her. Thus we can only have $3$ as the marginal dictator in the $2,3$-marginal mechanism at any $\succsim_{1}^{*},\succsim_{4}',\dots, \succsim_{N}'$. Since $2$ cannot affect $3$'s option set, we get eqation \ref{eq8}. Repeating the same arguments with equations \ref{eq5} and \ref{eq7} show that 
\begin{equation} \label{eq9}
    g_{3}(\succsim_{1}^{**},\succsim_{2}',\succsim_{4}',\dots , \succsim_{N}') = N-\{3\} \text{ for all } \succsim_{2}',\succsim_{4}',\dots , \succsim_{N}'
\end{equation}
Finally, by comparing equations \ref{eq8} and \ref{eq9}, we get the desired result that $g_{3}(\succsim_{1}',\succsim_{2}',\succsim_{4}',\dots , \succsim_{N}') = N-\{3\}$  for all $ \succsim_{2}',\succsim_{4}',\dots , \succsim_{N}'$.

Lastly, we come to the case in which $1$ is the marginal dictator among $\{1,2,3,4\}$ at the profile $\succsim_{5}^{*},\dots, \succsim_{N}^{*}$. Our strategy will be to reduce this to the previous case by showing that for some $\succsim_{3}^{\dagger}\in P^{\uparrow}(4)$, $\succsim_{4}^{\dagger}\in P^{\uparrow}(3)$, that $1$ is also the marginal dictator among $\{1,2,5,\dots, N\}$. 

By Lemma \ref{roommates_lemma}, we have 
\begin{equation}\label{eq10}
    g_{1}(\succsim_{2}',\succsim_{3}',\succsim_{4}',\succsim_{5}^{*},\dots, \succsim_{n}^{*})\supset \{2,3,4\}\text{ for all } \succsim_{2}',\succsim_{3}',\succsim_{4}'
\end{equation}

\noindent Let $k\in 5,\dots, N$. As a first step, we want to show that $k\in g_{1}(\succsim_{2}',\succsim_{3}',\succsim_{4}',\succsim_{5}^{*},\dots, \succsim_{n}^{*})\text{ for all } \succsim_{2}',\succsim_{3}',\succsim_{4}'$ and to do this it suffices to demonstrate a single preference profile $(\succsim_{2}'',\succsim_{3}'',\succsim_{4}'',\succsim_{5}^{*},\dots, \succsim_{n}^{*})$ where this holds, since $1$ is the marginal dictator among $\{1,2,3,4\}$. Let $\succsim_{3}$ top-rank $k$. Then $f$ matches $1$ and $2$ and also $3$ and $k$ at the profile $(\succsim_{1}^{*},\succsim_{2}^{*},\succsim_{3},\succsim_{4}',\succsim_{5}^{*},\dots, \succsim_{N}^{*})$ for any $\succsim_{4}'\in P^{\uparrow}(\{3,\dots, N\}).$ Let $\succsim_{2}^{**}$ be the same as $\succsim_{2}^{*}$, except that it top-ranks $3$ and let $\succsim_{3}^{**}$ be the same as $\succsim_{3}$, except that it top-ranks $2$. Since $1$ is the marginal dictator among $\{1,2,3,4\}$, $1$ and $2$ are still matched at the profile $(\succsim_{1}^{*},\succsim_{2}^{**},\succsim_{3}^{**},\succsim_{4}',\succsim_{5}^{*},\dots, \succsim_{N}^{*})$, so by Maskin monotonicity, we have

$$f(\succsim_{1}^{*},\succsim_{2}^{**},\succsim_{3}^{**},\succsim_{4}', \succsim_{5}^{*}, \dots , \succsim_{N}^{*})=f(\succsim_{1}^{*},\succsim_{2}^{*},\succsim_{3}^{*},\succsim_{4}', \succsim_{5}^{*}, \dots , \succsim_{N}^{*})$$
and in particular, $3$ and $k$ are still matched. Now consider the $1,k$-marginal mechanism at $(\succsim_{2}^{**},\succsim_{3}^{**},\succsim_{4}', \succsim_{5}^{*}, \dots , \succsim_{k-1}^{*},\succsim_{k+1}^{*}, \dots, \succsim_{N}^{*})$. Let $\succsim_{1}^{**}$ be the same as $\succsim_{1}^{*}$, except that it top-ranks $k$ and let $\succsim_{k}^{**}$ be the same as $\succsim_{k}^{*}$, except that it top-ranks $1$. We must have that $1$ and $k$ are matched in the marginal mechanism at  $(\succsim_{1}^{**},\succsim_{k}^{**})$, since otherwise Maskin monotonicity  says that $f$ gives the same result as though they had announced $(\succsim_{1}^{*},\succsim_{k}^{*})$, but in this case, $1$ and $2$ are matched and $3$ and $k$ are matched which is inefficient since we could swap $1$ and $3$'s  matches. Thus $(k,1)$ is in $I^{1,k}(\succsim_{2}^{**},\succsim_{3}^{**},\succsim_{4}', \succsim_{5}^{*}, \dots , \succsim_{k-1}^{*},\succsim_{k+1}^{*}, \dots, \succsim_{N}^{*})$. From the considerations above, there are four possibilities for this mechanism. However, since both $(2,3)$ and $(k,1)$ are in the marginal option set, the marginal mechanism is not constant. Note also that if $1$ top ranks $3$ and $k$ announces $\succsim_{k}^{*}$, then by equation \ref{eq10}, $1$ and $3$ are matched. Thus it is possible for both $1$ and $k$ to match with $3$ in this marginal mechanism. But since both can't match with $3$ at the same time, the marginal constraint is like the one shown on the right of figure \ref{roommates_1}, and there must be a single dictator. We will show that this dictator must be $1$. To do this, we will have to take a detour to the $3,k$-marginal mechanism.

By equation \ref{eq10}, $f_{1}(\succsim_{1}^{*},\succsim_{2}^{**},\succsim_{3}',\succsim_{4}',\succsim_{5}^{*},\dots, \succsim_{N}^{*})=2$ for all $\succsim_{3}',\succsim_{4}'$, so by Maskin monotonicity, we have $$f(\succsim_{1}^{*},\succsim_{2}^{**},\succsim_{3}',\succsim_{4}',\succsim_{5}^{*},\dots, \succsim_{N}^{*})=f(\succsim_{1}^{*},\succsim_{2}^{*},\succsim_{3}',\succsim_{4}',\succsim_{5}^{*},\dots, \succsim_{N}^{*})$$ for all $\succsim_{3}',\succsim_{4}'$. In particular, we have $g_{3}(\succsim_{1}^{*},\succsim_{2}^{**},\succsim_{4}',\succsim_{5}^{*},\dots, \succsim_{N}^{*})= \{4,\dots, N\}$ for all $\succsim_{4}'$ by equation \ref{eq3}. Consider the $3,k$-marginal mechanism at this profile.  If $3$ top ranks $k$ they are matched. If $3$ top ranks $4$ they are not. In the latter case, $k$ is matched to someone from $\{5,\dots, N\}$, which she prefers. Hence the marginal mechanism is either a dictatorship with $3$ as the dictator, or it is of the third type in which $3$ and $k$ are matched if either top-ranks the other. Let $\succsim_{3}''$ top rank $4$ and $\succsim_{k}''$ top rank $3$. In the type (3) marginal mechanism, we would have $3$ and $k$ matched in $$f(\succsim_{1}^{*},\succsim_{2}^{**},\succsim_{3}'',\succsim_{4}', \succsim_{5}^{*}, \dots \succsim_{k-1}^{*}, \succsim_{k}'',\succsim_{k+1}^{*}, \dots, \succsim_{N}^{*})$$
but then Maskin-monotonicity would imply that we get the same outcome if $2$ announced $\succsim_{2}^{*}$, yet at this profile, by equation \ref{eq3}, we would have $3$ matched to $4$. Hence we have that $3$ is the dictator in the $3,k$-marginal mechanism at $(\succsim_{1}^{*},\succsim_{2}^{**},\succsim_{4}',\succsim_{5}^{*},\dots \succsim_{k-1}^{*}, \succsim_{k+1}^{*}, \dots, \succsim_{N}^{*})$ for all $\succsim_{4}'$. This implies that $g_{3}(\succsim_{1}^{*},\succsim_{2}^{**},\succsim_{4}',\succsim_{5}^{*},\dots , \succsim_{k-1}^{*},\succsim_{k}', \succsim_{k+1}^{*}, \dots, \succsim_{N}^{*})= \{4,\dots, N\}$ for all $\succsim_{4}'$ and $\succsim_{k}'$. So we have $f_{3}(\succsim_{1}^{*},\succsim_{2}^{**},\succsim_{3}^{**},\succsim_{4}',\succsim_{5}^{*},\dots , \succsim_{k-1}^{*},\succsim_{k}^{**}, \succsim_{k+1}^{*}, \dots, \succsim_{N}^{*})=k$, and by nonbossiness
\begin{equation*}
f(\succsim_{1}^{*},\succsim_{2}^{**},\succsim_{3}^{**},\succsim_{4}',\succsim_{5}^{*},\dots , \succsim_{k-1}^{*},\succsim_{k}^{**}, \succsim_{k+1}^{*}, \dots, \succsim_{N}^{*})= \\
f(\succsim_{1}^{*},\succsim_{2}^{**},\succsim_{3}^{**},\succsim_{4}',\succsim_{5}^{*},\dots , \succsim_{k-1}^{*},\succsim_{k}^{*}, \succsim_{k+1}^{*}, \dots, \succsim_{N}^{*})
\end{equation*}
and on the right hand side we know that $1$ and $2$ are matched and $3$ and $k$ are matched. This implies that if $k$ switches from $\succsim_{k}^{*}$ to $\succsim_{k}^{**}$, $1$ and $k$ are not matched in the $1,k$-marginal mechanism at $(\succsim_{2}^{**},\succsim_{3}^{**},\succsim_{4}', \succsim_{5}^{*}, \dots , \succsim_{k-1}^{*},\succsim_{k+1}^{*}, \dots, \succsim_{N}^{*})$. Since either $1$ or $k$ must be the dictator in their marginal mechanism by earlier arguments, it must be $1$ and we have 
\begin{equation*}
    k\in g_{1}(\succsim_{2}^{**},\succsim_{3}^{**},\succsim_{4}', \succsim_{5}^{*}, \dots , \succsim_{N}^{*})
\end{equation*}
\noindent and since $2,3,4$ can't affect $1$'s option set we get 
\begin{equation*}
    k\in g_{1}(\succsim_{2}',\succsim_{3}',\succsim_{4}', \succsim_{5}^{*}, \dots , \succsim_{k-1}^{*},\succsim_{k+1}^{*}, \dots, \succsim_{N}^{*})
\end{equation*}
for all $\succsim_{2}',\succsim_{3}',\succsim_{4}'$. Since $k$ was arbitrary, together with equation \ref{eq10}, we have
\begin{equation}\label{eq11}
    g_{1}(\succsim_{2}',\succsim_{3}',\succsim_{4}', \succsim_{5}^{*}, \dots , \succsim_{N}^{*})=N-\{1\}
\end{equation}
for all $\succsim_{2}',\succsim_{3}',\succsim_{4}'$. This, however, gets us back to the first case since $1$ is the marginal dictator among $\{1,2,3,4\}$ at $\succsim_{5}^{*},\dots, \succsim_{N}^{*}$ and if $\succsim_{3}^{\dagger}\in P^{\uparrow}(4)$, $\succsim_{4}^{\dagger}\in P^{\uparrow}(3)$, then we must have a marginal dictator among $\{1,2,5,\dots, N\}$, however the only marginal dictator consistent with equation \ref{eq11} is agent $1$.
\qed

{\footnotesize
\bibliography{gspbib}
\bibliographystyle{econometrica}
}

\end{document}